\documentclass[a4paper,12pt]{amsart}
\usepackage[T1]{fontenc}
\usepackage[utf8]{inputenc}
\usepackage[margin=2.5cm]{geometry}

\usepackage{amsmath, amssymb,amsfonts}
\usepackage{thm-restate}
\usepackage{array}
\usepackage{multirow}
\usepackage[dvipsnames]{xcolor}
\usepackage{tikz}
\usetikzlibrary{shapes.misc}
\usetikzlibrary{shapes.geometric}
\usetikzlibrary {decorations.pathmorphing}

\tikzset{unseulcoin fill/.style={append after command={
   \pgfextra
        \draw[sharp corners, fill=#1, color=#1, line width = 0mm]% 
    (\tikzlastnode.west)% 
    [rounded corners=0pt] |- (\tikzlastnode.north)% 
    [rounded corners=0pt] -| (\tikzlastnode.east)% 
    [rounded corners=0pt] |- (\tikzlastnode.south)% 
    [rounded corners=5pt] -| (\tikzlastnode.west);
   \endpgfextra}}}

\tikzset{black node/.style = {draw=black,fill=black,circle, inner sep=0, minimum size=0.15cm}}
\usepackage{multicol}
\usepackage{paralist}
\usepackage{hyperref}
\hypersetup{
  pdftitle = {Long induced paths and forbidden patterns: Polylogarithmic bounds},
  pdfauthor = {J. Duron, L. Esperet, J.-F. Raymond},
  colorlinks = true,
  linkcolor = black!30!red,
  citecolor = black!30!green
}
\usepackage{cleveref}
\usepackage{dirtytalk}

\newtheorem{theorem}{Theorem}[section]
\newtheorem{corollary}[theorem]{Corollary}
\newtheorem{definition}[theorem]{Definition}

\newtheorem{lemma}[theorem]{Lemma}

\newtheorem{observation}[theorem]{Observation}
\newtheorem{remark}[theorem]{Remark}
\newtheorem{claim}[theorem]{Claim}

\def\cqedsymbol{\ifmmode$\lrcorner$\else{\unskip\nobreak\hfil
\penalty50\hskip1em\null\nobreak\hfil$\lrcorner$
\parfillskip=0pt\finalhyphendemerits=0\endgraf}\fi} 
\newcommand{\cqed}{\renewcommand{\qed}{\cqedsymbol}}

\newcommand{\N}{\mathbb{N}}

\newcommand{\intv}[2]{\left \{ #1,\dots, #2\right \}}

\newcommand{\Prop}{{\rm Prop}}

\newcommand{\cN}{\mathcal{N}}

\DeclareMathOperator{\polylog}{polylog}

\DeclareMathOperator{\rank}{{\sf rank}}
\DeclareMathOperator{\ribs}{{\sf ribs}}
\DeclareMathOperator{\depth}{{\sf depth}}

\newcommand{\pow}{*}

\newcommand{\floor}[1]{\left \lfloor #1 \right \rfloor}

\newcommand{\cst}{8}

\let\le\leqslant
\let\ge\geqslant
\let\leq\leqslant
\let\geq\geqslant

\title{Long induced paths and forbidden patterns: Polylogarithmic bounds}

\author[J.~Duron]{Julien Duron}
\address[J.~Duron]{Univ.\ Lyon, CNRS, ENS de Lyon, Université Claude Bernard Lyon 1, LIP UMR5668,
  Lyon, France}
\email{julien.duron@ens-lyon.fr}

\author[L.~Esperet]{Louis Esperet}
\address[L.~Esperet]{Univ.\ Grenoble Alpes, CNRS, Laboratoire G-SCOP,
  Grenoble, France}
\email{louis.esperet@grenoble-inp.fr}

\author[J.-F.~Raymond]{Jean-Florent Raymond}
\address[J.-F.~Raymond]{Univ.\ Lyon, CNRS, ENS de Lyon, Université Claude Bernard Lyon 1, LIP UMR5668,
  Lyon, France}
\email{jean-florent.raymond@cnrs.fr}

\date{\today}

\thanks{The authors are partially supported by the French ANR Projects TWIN-WIDTH
  (ANR-21-CE48-0014-01) and GRALMECO (ANR-21-CE48-0004), and by LabEx
  PERSYVAL-lab (ANR-11-LABX-0025).}

\begin{document}

\maketitle
\begin{abstract}
Consider a graph $G$ with a long path $P$. When is it the case that $G$ also contains a long induced path? This question has been investigated in general as well as within a number of different graph classes since the 80s. We have recently observed in a companion paper (\emph{Long induced paths in sparse graphs and graphs with forbidden patterns}, \texttt{arXiv:2411.08685}, 2024) that most existing results can be recovered in a simple way by considering forbidden ordered patterns of edges along the path $P$. In particular, we proved that if we forbid some fixed ordered matching along a path of order $n$ in a graph $G$, then $G$ must contain an induced path of order $(\log n)^{\Omega(1)}$. Moreover, we completely characterized the forbidden ordered patterns forcing the existence of an induced path of polynomial size. 

The purpose of the present paper is to completely characterize the ordered patterns $H$ such that forbidding $H$ along a path $P$ of order $n$ implies the existence of an induced path of order $(\log n)^{\Omega(1)}$.  These patterns are star forests with some specific ordering, which we call \emph{constellations}.  

As a direct consequence of our result, we show that if a graph $G$ has a path of length $n$ and  does not contain $K_t$ as a topological minor, then $G$ contains an induced path of order $(\log n)^{\Omega(1/t \log^2 t)}$. The previously best known bound was $(\log n)^{f(t)}$ for some unspecified function $f$ depending on the  Topological Minor Structure Theorem of Grohe and Marx (2015). 
\end{abstract}

\section{Introduction}

Consider a graph $G$ with a long path $P$. When is it the case that $G$ also contains a long induced path? Complete bipartite graphs must be forbidden as subgraphs, since these graphs have long paths but no induced paths of order 3. A classical result of Galvin, Rival, and Sands \cite{galvin1982ramsey} states that if $G$ contains an $n$-vertex path and is $K_{t,t}$-subgraph free, then it contains an induced path of order $\Omega((\log \log \log n)^{1/3})$. In the companion paper \cite{I}, we recently improved this bound to $\Omega\big( (\frac{\log \log n}{\log \log \log n})^{1/5} \big)$. This was further improved to $\Omega\big( \frac{\log \log n}{\log \log \log n}\big)$ by Hunter, Milojevi\'c, Sudakov, and Tomon in \cite{HMST}.

\smallskip

This question has also been investigated extensively when $G$ belongs to a specific graph class, such as outerplanar graphs,   planar graphs and graphs of bounded genus \cite{arocha2000long,esperet2017long,GLM16}, graphs of bounded pathwidth or treewidth \cite{esperet2017long,hilaire2023long}, degenerate graphs \cite{defrain2024sparse,nevsetvril2012sparsity}, and graphs excluding a minor or topological minor \cite{hilaire2023long}. 

In \cite{I}, we recently observed that most of the known results on this question can be recovered in a simple and unified way by considering the following variant of the problem. Consider that the vertices of the $n$-vertex path $P$ in $G$ are ordered following their occurrence in $P$. What forbidden ordered subgraphs in $G-E(P)$ force the existence of a long induced path in $G$? In \cite{I}, we showed that forbidding ordered matchings yields an induced path of order $n^{\Omega(1)}$ or $(\log n)^{\Omega(1)}$, depending on the structure of the matching. This was enough to imply all previously known results in the area, except a result of \cite{hilaire2023long} on graphs excluding a topological minor. 

As an example, Figure \ref{fig:planar}, left, depicts an ordered matching which is forbidden in any planar graph, in the sense above.  Here and in the remainder of the paper, the vertices are ordered from left to right on the figure. To see why such an ordered matching is forbidden, observe that adding the edges of the path $P$, we obtain a $K_{3,3}$-minor (see Figure \ref{fig:planar}, right), which contradicts planarity. A number of other natural graph classes, including graphs of bounded genus, treewidth, or pathwidth, similarly avoid some ordered matching, in the sense above. Therefore, any such graph with a path of order $n$ contains an induced path of order $(\log n)^{\Omega(1)}$. 

\medskip

 \begin{figure}[htb]
  \centering
    \includegraphics[scale=1.5]{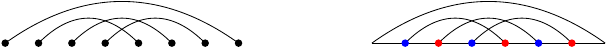}
  \caption{A forbidden ordered matching in planar graphs.\label{fig:planar}}
\end{figure}

In this paper, we completely characterize the forbidden ordered subgraphs yielding induced paths of order  $(\log n)^{\Omega(1)}$. These graphs are star forests with a specific vertex ordering, which we call \emph{constellations}. Our proof has two parts: we first show that forbidding a constellation yields induced paths of order  $\polylog(n)$, and we then construct a graph without any of these constellations, which has no induced path of order $\Omega((\log \log n)^2)$. The construction is inspired by the recent construction of \cite{defrain2024sparse} of a 2-degenerate Hamiltonian $n$-vertex graph without induced path of order $\Omega((\log \log n)^2)$.

As a direct consequence of our main result, we obtain that graphs which do not contain $K_t$ as a topological minor, and which contain a path on $n$ vertices, also contain an induced path of order $(\log n)^{\Omega(1/t\log^2 t)}$. This simplifies and improves upon an earlier result of \cite{hilaire2023long}, in which the exponent was an unspecified function of $t$ (relying on structure theorems of Robertson and Seymour, and Grohe and Marx). In the particular case of forbidden $K_t$-minors, this also improves upon the bound of order $(\log n)^{\Omega(1/t^2)}$ obtained in \cite{I} using forbidden ordered matchings (we note that the proof of the weaker bound in \cite{I} is significantly simpler than the proof of the stronger bound obtained in the present paper).

\subsection*{Organization of the paper}
In \Cref{sec:prelim} we introduce the necessary tools and definitions. 
Our main result, a proof that graphs with long paths and without constellations have long induced paths (Theorem \ref{th:peel_omega}), is proved in Section  \ref{sec:stars}. The construction showing that constellations are the only ordered subgraphs whose avoidance yields induced paths of polylogarithmic size is given in Section \ref{sec:loglogsqr}.
We conclude with some additional remarks in \Cref{sec:open}.

\section{Preliminaries}
\label{sec:prelim}

In all the paper, to avoid any ambiguity we always consider the \emph{order} of a path $P$ (its number of vertices), denoted by $|P|$. We never refer to the length (number of edges) of a path. 

\medskip

Logarithms are in base 2. As we will often be using several levels of exponentiations, it will sometimes be more convenient to write exponentiation in-line: we will then write $a\pow b$ instead of $a^b$. We omit parentheses for the sake of readability but it should be implicit that $\pow$ is not associative and $n_1 \pow n_2\pow \cdots  \pow n_k=n_1 \pow (n_2 \pow (\cdots \pow n_k )\cdots ) $. Similarly, we use $\log^{(i)}$ to denote the logarithm iterated $i$ times. That is, $\log^{(0)}$ is the identity function and for every integer $i\geq 1$ and every $x$ s.t.\ $\log^{(i-1)}(x)>0$, $\log^{(i)} x = \log (\log^{(i-1)} x)$.

\subsection*{Forbidden patterns}

An \emph{ordered graph} is a graph with a total order on its
vertex set. Consider an ordered graph $G$ with order
$v_1,\ldots,v_n$ and an ordered graph $H$ with order
$u_1,\ldots,u_k$. We say that $G$ \emph{contains $H$ as an ordered subgraph}
% (or that \emph{$G$ contains the pattern $H$})
if
there exist $1\le a_1<a_2<\ldots <a_k\le n$ such that for all $1\le
i,j\le k$, if $u_{i}$ is adjacent to $u_{j}$ in $H$, then $v_{a_i}$ is
adjacent to $v_{a_j}$ in $G$. In othe words, $H$ appears as a subgraph in
$G$ in such a way that  the ordering of the copy of $H$ in  $G$ is consistent with the
ordering of $G$.

 \medskip

Let $G$ be a graph and $P=v_1,v_2,\ldots,v_n$ be a Hamiltonian path in
$G$. Note that $P$ allows us to consider $G$ and $G-E(P)$ (the spanning
subgraph of $G$ obtained by removing the edges of $P$) as ordered
graphs, that is with $v_i\prec v_j$ if and only if $i<j$. Given an ordered graph $H$, we say that \emph{$(G,P)$ contains $H$ as a
pattern} if the ordered graph $G-E(P)$ contains $H$ as an ordered
subgraph. If $(G,P)$ does not contain $H$ as a pattern, we say that \emph{$(G,P)$
  avoids the pattern $H$}. When $P$ is clear from the context we simply say that $G$ contains or avoids the pattern $H$ (but in all such instances we really mean that $H$ is a pattern with respect to some Hamiltonian path $P$, so the edges of $P$ are not part of the pattern). 
  
  \smallskip
  
  In \cite{I} we studied the function $g_H(n)$ defined as the maximum integer $k$ such that for every graph $G$ with a path $P$ of order $n$ that avoids $H$ as a pattern, $G$ has an induced path of order at least $k$.
Observe that we can assume  that $P$ is a Hamiltonian path in $G$ (by considering the subgraph of $G$ induced by $P$ instead of $G$).
If $H=K_2$ then $G$ is precisely an induced path on $n$ vertices, so $g_{K_2}(n)=n$. 

\medskip

Let $A$ and $B$ be two ordered graphs. The \emph{concatenation} of $A$
and $B$, denoted by $A\cdot B$, is the ordered graph whose graph is the
disjoint union of $A$ and $B$, and whose order consists of the ordered
vertices of $A$, followed by the ordered vertices of $B$. We will need the following simple result, proved in \cite{I}. 

\begin{lemma}[\cite{I}]\label{lem:concat}
Let $A$ and $B$ be two ordered graphs. Then, for any integer $n\ge 0$,
\[g_{A\cdot B}(n)\ge \min\{g_A(\lfloor n/2\rfloor),g_B(\lceil n/2\rceil)\}.\]
\end{lemma}

We proved in \cite{I} that if $g_H(n)=\omega(\log n)$, then $H$ must be a matching. Hence, better-than-logarithmic bounds on the size of induced paths can only be obtained by considering very simple patterns, namely ordered matchings. It is thus natural to investigate  $g_H(n)$ when $H$ is an ordered matching. In this case, we proved the following in \cite{I}: \begin{itemize}
    \item either $H$ is \emph{non-crossing} (that is, it does not contain vertices $a<b<c<d$ with edges $ac,bd$) and then $g_H(n) = n^{\Theta(1)}$, or
    \item $H$ contains a pair of crossing edges and then $g_H(n) = (\log n)^{\Theta(1)}$.
\end{itemize} 

\begin{figure}[htb]
  \centering
    \includegraphics[scale=1.5]{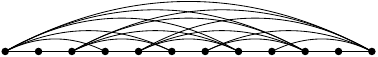}
  \caption{A graph with no long induced path.\label{fig:construction}}
\end{figure}

We also studied the construction illustrated in Figure \ref{fig:construction}. In this graph, the $i$-th vertex on the path $P$ is adjacent to the $j$-th vertex whenever $i$ is odd, $j$ is even, and $i<j$. It can be checked that every induced path has order at most 4, independently of the order of the path $P$. Note that every vertex has all its neighbors preceding it or all its neighbors succeeding it, except its immediate predecessor and successor on the path $P$. This implies the following.

\begin{observation}[\cite{I}]\label{obs:d2a}
Let $H$ be an ordered graph such that $\{g_H(n): n\in \mathbb{N}\}$ is unbounded. Then, for each vertex $v\in V(H)$, all neighbors of $v$ are 
predecessors of  $v$, or all neighbors of $v$ are  successors of $v$. In particular, $H$ is bipartite. 
\end{observation}

\section{Constellations}
\label{sec:stars}

\subsection{Definitions and main properties}

The \emph{$r$-star} is the complete bipartite graph $K_{1,r}$. We say that the vertex of degree $r$ is the \emph{center of the star} (if $r=1$, both endpoints can be the center of the star, and otherwise the center is unique). Recall that by Observation~\ref{obs:d2a}, for $g_H$ to be unbounded, an ordered graph $H$ must have the property that every vertex  is larger than all its neighbors, or smaller than all its neighbors. So, we only need to consider two orderings of a star: the \emph{right star} where the center is the smallest vertex in the ordering, and the \emph{left star} where the center is the largest vertex (we consider that a 1-star is a left star and a right star).
In the course of showing that traceable graphs of bounded degeneracy have long induced paths, Ne{\v{s}}et{\v{r}}il, and Ossona de Mendez proved a lemma that can be restated as follows in terms of excluded patterns.

\begin{lemma}[{\cite[Lemma~6.3]{nevsetvril2012sparsity}}]\label{lem:nespom}
If $H$ is a left or right $r$-star, then
$g_{H}(n) \geq \frac{\log ((r-1)n +1)}{\log r}$.
\end{lemma}

This shows that avoiding a single right or left star as a pattern guarantees the existence of an induced path of logarithmic order. In the following, our goal will be to obtain polylogarithmic bounds for  patterns consisting of a constant number of constant size stars. 

\smallskip

We now introduce  constellations, a particular type of ordered star forests. A \emph{constellation} $H$ consists of a disjoint union of 
  stars $S_1,\ldots,S_t$, each of which is a left or right star,
  and is defined inductively as follows:
 \begin{itemize}
     \item either the center of one of the stars, say $S_1$, is the first vertex of $H$, and  $H-S_1$ is a constellation,
     \item or the center of one of the stars, say $S_t$, is the last vertex of $H$, and  $H-S_t$ is a constellation,
     \item or $H$ is the concatenation of two constellations.
 \end{itemize}
 
 In the first item above, $H$ is called a \emph{right constellation} (the star whose center is the first vertex of $H$ is a right star), and in the second item above,
 $H$ is called a \emph{left constellation} (the star whose center is the last vertex of $H$ is a left star). We emphasize that the three items in the definition of a constellation are not mutually exclusive: for instance the concatenation of a right star and a left star is a constellation that satisfies all three items. Note also that any ordered matching is a left constellation and a right constellation. See Figure \ref{fig:example} for an example of a constellation. The constellation on the figure is the concatenation of a left constellation (on the left-hand side) and a right constellation (consisting of two intertwined stars, on the right-hand side of the figure). 
 
 \medskip
 
 \begin{figure}[htb]
  \centering
    \includegraphics[scale=1.5]{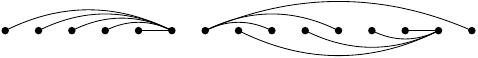}
  \caption{A constellation.\label{fig:example}}
\end{figure}

%First for a given ordered graph $H$ and an ordered subgraph $H'$ of $H$, the \emph{span of $H'$ in $H$}, denoted by $\Span_H(H')$, is the inclusion-wise minimal set of consecutive (in the ordering) vertices of $H$ that contains $V(H')$.

%An ordered graph is a \emph{$(t,r)$-constellation} if it consists of a union of $t$ vertex-disjoint left or right $r$-stars $S_1, S_2, \ldots, S_t$ where some star, say $S_1$, does not have its center in $\bigcup_{i \ge 2} \Span_H(S_i)$, and the ordered subgraph $S_2\cup S_3 \dots \cup S_t$ is a $(t-1, r)$-constellation.

\medskip

Although it is not needed in this section, the following equivalent definition of constellations will be useful in \Cref{sec:loglogsqr}, in the proof of \Cref{th:loglogbis}. In an ordered graph $G$, we say that a vertex $v$ is \emph{outside} an induced subgraph $H$ of $G$ if $v$ precedes or succeeds the vertex set of $H$. 

\begin{lemma}\label{lem:loccons}
Let $H$ be an ordered star forest consisting only of left  and right stars. Then, $H$ is a constellation if and only if 
\begin{itemize}
    \item[$(\star)$] There is an ordering $S_1,\ldots,S_t$ of the stars of $H$ such that for any $i<j$, the center of $S_i$ is outside $S_j$.
\end{itemize}
\end{lemma}

\begin{proof}
Assume first that $H$ is a constellation. 
We prove $(\star)$ by induction on the number of stars of $H$. If $H$ is a concatenation $H_1\cdot H_2$ of two non-trivial constellations, then $(\star)$ certainly holds by induction on $H_1$ and $H_2$, as each star in $H_1$ precedes each star in $H_2$. So, we can assume that $H$ is a left or right constellation, say a right constellation by symmetry. Then, $H$ contains a star $S_1$ whose center $c$ is the first vertex of $H$, and thus the result follows by induction on the constellation $H-S_1$ (as $c$ precedes all stars of $H-S_1$).

Assume now that $(\star)$ holds, and consider the first star $S_1$ in the order given by $(\star)$. Assume by symmetry that $S_1$ is a right star. If the center $c_1$ of $S_1$ precedes all the other stars of $H$ then $c_1$ is the first vertex of $H$, and thus $H$ is a right constellation (by induction on $H-S_1$).  Otherwise, let $H_1$ be the subgraph of $H$ induced by the stars $S_i$ which precede $c_1$, and let $H_2=H-H_1$. By assumption, $H_1$ and $H_2$ are non-empty, and since $(\star)$ is closed under taking a subset of the stars, both $H_1$ and $H_2$ satisfy $(\star)$. It remains to observe that $H=H_1\cdot H_2$, by definition of $S_1$ and $c_1$, so by induction $H$ is a concatenation of constellations, and therefore also a constellation.
\end{proof}

A constellation consisting of $t$ stars, each of which is an $r$-star, is called a \emph{$(t,r)$-constellation}.

\subsection{Main result}

Our first main result is a polylogarithmic lower bound on $g_H(n)$ when $H$ is a constellation.

\begin{theorem}\label{th:peel_omega}
There exists a constant $\mu > 0$ such that for any integers $r\ge 1$ and $t \ge 1$, and any $(t,r)$-constellation $H$,
\[
    g_{H}(n) \geq (\log_{r} n)^{\frac{\mu}{t(\log t)^2}}.
\]
\end{theorem}

Let $G$ be an ordered graph with order $v_1,\ldots,v_n$, and let $H$
be an ordered subgraph of $G$ with vertex set $v_{a[1]}, v_{a[2]},\ldots,
v_{a[k]}$ (for $1\le a[1] < a[2] < \ldots < a[k] \le n$). The \emph{gap}
of $H$ in $G$ is defined as the minimum of $a[{i+1}]-a[i]$, for $1\le i
\le k-1$. The definition of the gap naturally extends to patterns in
pairs $(G,P)$ where $G$ is a graph and $P$ a Hamiltonian path in
$G$. 

\medskip

\Cref{th:peel_omega} is a consequence of \Cref{th:peel} below which, informally, states that a graph either contains a ``large scale'' version of a constellation as a pattern, or contains a ``long''\footnote{The inequalities describing how ``large'' should be related to ``long'' for our argument to work are gathered in \Cref{lem:bounds} hereafter.} increasing induced path. This is illustrated by the following simplified form of \Cref{th:peel}.

\begin{theorem}[simplified form of \Cref{th:peel}]
Let $H$ be a $(t,r)$-constellation.
There are functions
\[f(n) = (\log n)^{\Theta_{r}(1/(t(\log t)^2))}\quad \text{and}\quad g(n) = \frac{n}{2 \pow (\log n) \pow \big( 1 - \Theta_{r}(1/(t(\log t)^2)) \big)}\]
such that for every graph $G$ with a Hamiltonian path $P=v_1,\ldots,v_n$, either $(G,P)$ contains the pattern $H$ with gap $g(n)$, or $G$ contains an induced path of order at least $f(n)$ which is increasing with respect to the order $v_1,\ldots,v_n$.
\end{theorem}

For the purpose of the induction  we need to define the aforementioned functions $f$ and $g$ with extra parameters, as well as additional functions, which we do now. In the rest of the section, $r \geq 1$ is a fixed integer. Most functions we introduce depend on $r$ implicitly, but as $r$ is fixed, we do not consider them explicitly as functions of $r$.

\begin{definition}\label{def:phigamma}
Let $\varphi, \eta, \gamma \colon \N \cup \{-1\} \to (0,1)$ be three functions such that for every $t\in \N$ we have:
\begin{align}
    \gamma(t) - \gamma(t-1) &\geq \cst \cdot \varphi(t-1),\\
    \label{eq:unmoins} 1 - \gamma(t-1) &\geq \cst \cdot \varphi(t-1),\\
    \varphi(t-1) &> \eta(t) > \varphi(t),\ \text{and}\\
    \varphi(t-1) - \eta(t) &> \varphi(t) - \eta(t+1).
\end{align}
\end{definition}

\begin{remark}\label{rem:exist-phigamma}
Functions as in \Cref{def:phigamma} exist, for instance for every $t\in \N \cup \{-1\}$ we could take
\begin{align*}
\varphi(t) &:= \frac{1}{\cst}\cdot \frac{1}{\alpha }\cdot  \frac{1}{(t+10)(\log (t+10))^2}, \quad \text{where}\ \alpha:=\sum_{i=-1}^{\infty}\frac{1}{(i+10)(\log (i+10))^2}\approx 0.22,\\
\eta(t) &:= \frac{\varphi(t-1) + \varphi(t)}{2},\quad and\\
\gamma(-1) &:= 0\quad \text{and if}\ t\geq 0,\ \gamma(t) := \cst \cdot \sum_{i=-1}^{t-1} \varphi(i).
\end{align*}
Actually we could replace $(t+10) (\log(t+10))^2$ above by any function of the form \[\Theta(t (\log t)(\log \log t) \cdots (\log \cdots \log t)^2),\] where the square is only on the last factor. Indeed, by the Cauchy Condensation Test (see \cite{Mor38}), for any such function $\rho$ and $t_0\in \N$ large enough so that $1/\rho(t_0)$ is defined, the series $\sum_{t=t_0}^\infty 1/\rho(t)$ converges.

The ``${}+ 10$'' term above is only here to ensure that the functions are indeed defined for small values.
\end{remark}
%Observe that $\varphi\colon \N \to \R_{\geq 0}$ is a decreasing function such that $\lim_{n\to \infty} \sum_{i=1}^n \varphi(i) = \frac{1}{\cst}$.

\begin{definition}\label{def:functions}
We use the functions of \Cref{def:phigamma} above to define, for all integers $n\geq 1$, $t\geq 1$, $p\geq 0$, the following functions:

\begin{align*}
    f(n, t, p) & := (\log_{r+1} n)^{\varphi(t)} - p/2 - 4^{\frac{1}{ \varphi(t-1) - \eta(t)}},\\
    h(n, t, p) & := (\log_{r+1} n)^{\eta(t)} + p/2 - 4^{\frac{1}{ \varphi(t-2) - \eta(t-1)}},\\
    g(n, t, p) & := \frac{n}{(6(r+1)) \pow \left(2(\log_{r+1} n)^{\gamma(t)} \cdot (3(\log_{r+1} n)^{\varphi(t)} - p) \right)},\ \text{and}\\
    s(n, t, p) & := \frac{g(n/3, t - 1, p)-1}{2r + 1}.
\end{align*}
\end{definition}

The properties of the functions $f$, $g$, $h$, and $s$ defined above that are crucial for our proof are given in \Cref{lem:mono} and \Cref{lem:bounds} below. The proofs of these properties are a sequence of tedious and relatively unexciting computations, so we defer them to \Cref{sec:stars-proofs}.

\begin{restatable}{lemma}{lemmono}\label{lem:mono}
For any integers $p\geq 0$, $t\geq 1$, and $n$ such that $\log_{r+1} n \ge 4 \pow \frac{1}{\varphi(t) \cdot(\varphi(t-1) - \eta(t))}$ and $p\leq 2\cdot  (\log_{r+1} n)^{\varphi(t)}$, we have 
\[
f(n, t-1, p) \geq f(n, t, p),\ 
g(n, t-1, p) \geq g(n, t, p),\ \text{and}\ 
h(n, t-1, p) \geq h(n, t, p).
\]
\end{restatable}

\begin{restatable}{lemma}{lembounds}\label{lem:bounds}
For any integers $r\ge 1, t\ge 1, p\ge 0, n\ge 1$, such that %
\begin{align}
    \label{cond:bign} \log_{r+1} n &\geqslant (2 + p/2)^{1/\varphi(t)} + 4^{\frac{1}{\varphi(t) \cdot (\varphi(t-1) - \eta(t))}} \quad and\\
    \label{cond:smallp} p &< 2 (\log_{r+1} n)^{\varphi(t)},
\end{align}
\noindent we have the following inequalities:
\begin{align}
    \label{eq:recursionf} f\big(s(n, t, p), t, p+1\big) &\geq f(n, t, p) - 1,\\
    \label{eq:recursionh} h\big(s(n, t, p), t, p+1\big) &\geq h(n, t, p),\\
    \label{eq:recursiong} g\big(s(n, t, p), t, p+1\big) &\geq g(n, t, p),\\
    \label{eq:middlef} f(n/3, t-1, p) &\geq h(n, t, p),\ \text{and}\\
    \label{eq:middleg} s(n,t,p) &\geq g(n, t, p).
\end{align}
\end{restatable}

\medskip

Let $G$ be a graph with a Hamiltonian path $P = v_1, \dots, v_n$, $n\geq 2$. For any two integers $a$ and $b$ such that $1\le a\le b \le n$, we denote by $G[a,b]$ the ordered subgraph of $G$ induced by the vertices $v_a, \ldots,v_b$.
Let $a_1, \dots, a_d$ denote the indices of the neighbors of $v_1$ in $G$ (in the same order as in $P$, so $a_1=2$) and let $a_{d+1}=n$.
The \emph{stretch} of $G$ is defined as $\max_{i\in \intv{1}{d} } a_{i+1}-a_i$.
Let $i \in \intv{0}{d}$ be the minimum index maximizing $a_{i+1}-a_{i}$ and call the ordered subgraph $G[a_{i},a_{i+1}-1]$ of $G$ the \emph{successor} of $G$.

\medskip

\begin{lemma}\label{lem:stretch}
Let $G$ be a graph with a Hamiltonian path $P=v_1, \dots ,v_n$ and let $s,m\in \N_{\geq 1}$ % s for stretch, t for threshold
with $m < n$. 
If for every $i,j \in \intv{1}{n}$ such that $j-i+1 \geq n/m$, $G[i,j]$ has stretch at least $\frac{j - i + 1}{s}$ then $G$ has an increasing induced path of length at least $\frac{\log m}{\log s}$ starting from $v_1$. 
\end{lemma}

\begin{proof}
Let $G_0 = G$. For every $i\ge 0$ and as long as $|G_i|\geq n / m$, we define $G_{i+1}$ as the successor of $G_i$. Let $p$ be the index of the last graph defined that way. For every $i \in \intv{1}{p}$, let $v_{a[i]}$ be the first vertex of $G_i$. Clearly $G_{i+1}$ is an (ordered) induced subgraph of $G_i$ and $v_{a[i]}$ has only one neighbor in $G_{i+1}$, that is $v_{a[i+1]}$. So, $v_{a[0]}, \dots, v_{a[p]}$ is an induced path.

By definition of $p$, for every $i\in \intv{0}{p-1}$, $|G_i| \geq n/m$ so by assumption $G_i$ has stretch at least $|G_i|/s$. Hence, $G_{i+1}\geq |G_i|/s \geq n/s^{i+1}$.
Recall that $|G_p| < n/m$. So, $n/s^p \leq n/m$ and $p \geq \frac{\log m}{\log s}$.
\end{proof}

The following is the main technical result of the section.

\begin{theorem}\label{th:peel}
Let $r\ge 1$ be a fixed integer, and let $f$ and $g$ be the functions introduced in Definition~\ref{def:functions}. Let $H$ be a $(t,r)$-constellation.
Let $G$ be a graph with a Hamiltonian path $P=v_1,\ldots,v_n$. Then, either $(G,P)$ contains the pattern $H$ with gap at least $g(n, t, 0)$, or $G$ contains an induced path of size at least $f(n, t, 0)$ which is increasing with respect to the order $v_1,\ldots,v_n$.
\end{theorem}

\begin{proof}
Recall that $r\ge 1$ is fixed and all the functions of Definition \ref{def:functions} implicitly depend on~$r$.
Recall also that by definition, $H$ is either a right constellation, a left constellation, or a concatenation of smaller constellations.
For the sake of induction we will actually prove the following stronger proposition $\Prop(t, n, p)$ for any integers $t\ge 1$, $n\ge 1$, and $p\ge 0$:
\begin{description}
    \item[$\Prop(n, t, p)$] 
One of the following holds
\begin{enumerate}[(P1)]
    \item \label{it:first} $H$  is a right (resp.\ left) $(t,r)$-constellation and $G$ contains an increasing induced path of order $f(n, t, p)$ starting at the first (resp.\ ending at the last) vertex, or
    \item \label{it:induced} $G$ contains an increasing induced path of order $h(n, t, p)$, or
    \item \label{it:pattern} $G$ contains the pattern $H$ with gap at least $g(n, t, p)$.
\end{enumerate}
\end{description}

\noindent {\bf Base case ($t = 1$).}
\newcommand{\Pa}{(\hyperref[it:first]{P1})}
\newcommand{\Pb}{(\hyperref[it:induced]{P2})}
\newcommand{\Pc}{(\hyperref[it:pattern]{P3})}

When $t = 1$, $H$ is a right or left star (and in particular, a left or right constellation). By symmetry, we can assume that $H$ is a right star. We will prove that either \Pa{} or \Pc{} holds in this case. To do so, we call \Cref{lem:stretch} with $s := 2r$ and $m := 2^{f(n, 1, p) \cdot \log 2r}$.

The first step is to prove that either \Pc{} holds, or for any indices $i$ and $j$ with $j - i + 1 \ge n/m$, the subgraph $G[i,j]$ has stretch at least $\frac{j - i + 1}{s}$.
Hence, assume that for such pair $i,j$, $G[i,j]$ has stretch less than $\frac{j - i + 1}{s} = \frac{j - i +1}{2r}$. Then, one finds the star $H$ with gap at least $\frac{n}{ms}$ as follows: take $v_i$, and a neighbor of $v_i$ in each interval $\left [i + (2k - 1)\frac{j - i + 1}{s}, i + 2k\frac{j - i + 1}{s}-1 \right ]$ for $1\le k\le r$.
But since
\begin{align*}
sm & = 2r \cdot 2^{f(n, 1, p) \cdot \log 2r}\\
& = 2r \cdot 2 * \left( \left( (\log_{r+1} n)^{\varphi(1)} - p/2 - 4^{\frac{1}{\varphi(0) - \eta(1)}} \right) \cdot \log 2r \right) &\text{by definition}\\
& \le 2r \cdot (2r) * \left( (\log_{r+1} n)^{\varphi(1)} - p/2 \right)\\
& \le  (2r) * \left( 2 \left( (\log_{r+1} n)^{\varphi(1)} - p/2 \right) \right) &\text{since $f(n, 1, p) \ge 1$} \\
& <  (6(r+1)) * \left( 6 \left( (\log_{r+1} n)^{\varphi(1)} - p/2 \right) \right)\\
& =  (6(r+1)) * \left( 6(\log_{r+1} n)^{\varphi(1)} - 3p  \right) = n/g(n, 1, p)
\end{align*}
we have $\frac{n}{ms} \ge g(n, 1, p)$, and so we found $H$ with gap at least $g(n, 1, p)$ in $G$, and proved \Pc{}.
Hence, we can assume that $G[i,j]$ has stretch at least $\frac{j - i + 1}{s}$ for any pair $i,j$ with $j-i+1\ge n/m$, and thus we can apply \Cref{lem:stretch}, finding a path that starts in $v_1$ of size $\frac{\log m}{\log s} = \frac{f(n, 1, p)\log 2r}{\log 2r} = f(n, 1, p)$ This proves that \Pa{} holds and concludes the proof of the base case $(t=1)$.

\medskip
\noindent {\bf Induction step ($t>1$).}
We distinguish two cases below, depending whether $H$ is a concatenation of smaller constellations or a left or right constellation. For the induction we will assume that that $t>1$ and that for every $t'<t$, and every $n'\geq 1$ and $p'\geq 0$, $\Prop(n',t',p')$ holds.

\medskip
\noindent {\it Case 1: $H$ is a concatenation of non-empty constellations.}

\smallskip

If $H$ is the concatenation of a $(t_1, r)$-constellation $H_1$ and  a $(t_2, r)$-constellation $H_2$ (with $t_1,t_2>0$ and $t_1+t_2=t$), then we apply induction on $G_1=G[1, \lceil n/3 \rceil]$ with pattern $H_1$, and induction on $G_2=G[\lfloor 2n/3 \rfloor,n]$ with pattern $H_2$. If, in one of them, the outcome is \Pa{}, then the resulting increasing induced path has order at least $f\big(n/3, \max(t_1, t_2), p\big)\ge f(n/3, t - 1, p)$, since $f$ is decreasing in $t$ (by \Cref{lem:mono}). If, in one of them, the outcome is \Pb{}, then the resulting increasing induced path has order at least $h\big(n/3, \max(t_1, t_2), p\big)\ge h(n/3, t - 1, p)\ge f(n/3, t - 1, p)$. It then follows from  \Cref*{lem:bounds}.\eqref{eq:middlef}, that in both cases this path has order at least  $h(n, t, p)$. Hence, \Pb\ holds for $G$. 

%Louis ici

Otherwise, both applications of the induction hypothesis result in \Pc{} for $G_1$ with pattern $H_1$ and $G_2$ with pattern $H_2$. That is, we find $H_1$ with gap $g(n/3, t_1, p)$ in $G[1,\lceil n/3 \rceil]$ and $H_2$ with gap $g(n/3, t_2, p)$ in $G[\lfloor 2n/3 \rfloor,n]$. Since $g$ is decreasing with $t$ (by \Cref{lem:mono}), in particular, we find $H_1$ and $H_2$ each with gap at least $g(n/3, t-1, p)$ which is at least $s(n, t, p)$, so at least $g(n, t, p)$ by \Cref*{lem:bounds}.\eqref{eq:middleg}. Furthermore, the patterns $H_1$ and $H_2$ are separated by at least $n/3 - 2 \ge g(n, t, p)$ vertices, so we have the pattern $H=H_1 \cdot H_2$ with gap at least $g(n, t, p)$. Hence, \Pc{} holds.

\medskip
\noindent {\it Case 2: $H$ is a left or right constellation.}

\smallskip

We will extensively use the following two claims.
\begin{claim}\label{cl:bigp}
If $p \geqslant 2 \cdot f(n, t, 0)$ then $\Prop(n,t,p)$ holds.
\end{claim}
\begin{proof}
Indeed in this case by \hyperref[def:functions]{definition} of $f$ we have 
$
f(n,t,p) \leqslant 0,
$
so \Pa{} is trivially satisfied.%
\cqed
\end{proof}
\begin{claim}\label{cl:smalln}
If $\log_{r+1} n < (2 + p/2)^{1/\varphi(t)} + 4^{\frac{1}{\varphi(t) \cdot(\varphi(t-1) - \eta(t))}}$ then $\Prop(n,t,p)$ holds.
\end{claim}
\begin{proof}
Indeed in this case 
\begin{align*}
f(n,t,p) & = \left(\log_{r+1} n\right)^{\varphi(t)} - p/2 - 4^{\frac{1}{ \varphi(t-1) - \eta(t)}}& \text{by definition}\\
& < \left(\left(2 + p/2\right)^{1/\varphi(t)} + 4^{\frac{1}{\varphi(t) \cdot( \varphi(t-1) - \eta(t))}}\right)^{\varphi(t)} - p/2 - 4^{\frac{1}{ \varphi(t-1) - \eta(t)}}\\
& \le \left( \left(2 + p/2\right)^{1/\varphi(t)}\right)^{\varphi(t)} + \left(4^{\frac{1}{\varphi(t) \cdot (\varphi(t-1) - \eta(t))}}\right)^{\varphi(t)} - p/2 - 4^{\frac{1}{ \varphi(t-1) - \eta(t)}}\\&&\hspace{-3cm} \text{by sub-additivity of $x \mapsto x^{\varphi(t)}$}\\
& = 2 + p/2 + 4^{\frac{1}{\varphi(t-1) - \eta(t)}} - p/2 - 4^{\frac{1}{ \varphi(t-1) - \eta(t)}}\\
& = 2.
\end{align*}
So \Pa{} is satisfied by any edge incident to the first vertex of $G$.%
\cqed
\end{proof}

By \Cref{cl:bigp} and \Cref{cl:smalln} we can assume without loss of generality that
\begin{equation}\label{eq:clim}
p < 2 \cdot f(n, t, 0)\ \text{and}\ \log_{r+1} n \geqslant (2 + p/2)^{1/\varphi(t)} + 4^{\frac{1}{\varphi(t) \cdot(\varphi(t-1) - \eta(t))}}.
\end{equation}

\medskip

 By symmetry, up to reversing $P$, we may assume without loss of generality that $H$ is a right constellation. Let $H=S_1,\ldots,S_t$. Since $H$ is a right $(t,r)$-constellation, the vertex of degree $r$ in $S_1$ is the first vertex of $H$. Let $H^-$ be the ordered graph obtained from $H$ by deleting~$S_1$.

\begin{claim}\label{claim:large-stretch-or-pattern}
Let $k$ be an integer such that $k \leq n/3  - 2$.
Suppose that $G\left [\floor{n/3}, \floor{2n/3} \right ]$ contains the pattern $H^-$ with gap at least $k$.
Then either $G$ has stretch more than $\frac{k - 1}{2r + 1}$,
or $G$ contains the pattern $H$ with gap at least $\frac{k-1}{2r + 1}$.
\end{claim}
\begin{proof}
Suppose that $G$ has stretch at most $\frac{k - 1}{2r + 1}$.
Then $v_1$ has at least one neighbor in any set of $\frac{k-1}{2r + 1}$ consecutive vertices of $G$.
Hence, for any two vertices $v_a$ and $v_b$ of $G$ such that $b - a \geq k$, $v_1$ has $2r + 1$ neighbors indexed $i_0, \ldots, i_{2r}$
such that for all $0\le j\le 2r$,  $a + 1 + j\frac{k - 1}{2r + 1} \leq i_j < a + 1 + (j+1) \frac{k - 1}{2r + 1}$.
In particular, by taking all $i_j$ with $j$ in $\{ 1, 3, \ldots, 2r - 1\}$, one finds a copy with gap at least  $\frac{k-1}{2r + 1}$ of a right $r$-star centered in $v_1$ and whose leaves have indices between $a + \frac{k-1}{2r + 1}$ and $b - \frac{k-1}{2r + 1}$.

As $G\left [\floor{n/3}, \floor{2n/3} \right ]$ contains the pattern $H^-$ with gap at least $k$, $G$ contains as a pattern with gap at least $\frac{k-1}{2r + 1}$ the ordered graph consisting of $H^-$ preceded by a right $(r\cdot (|V(H^-)| + 1))$-star, having at least $r$ leaves in each gap of the pattern $H^-$, together with $r$ leaves before
$v_{\floor{n/3}}$ and $r$ leaves after $v_{\floor{2n/3}}$. In particular, $G$ contains  the pattern $H$ with gap at least $\frac{k-1}{2r + 1}$, as desired.\cqed
\end{proof}

Let $G_{\text{mid}} := G[\lfloor n/3 \rfloor, \lceil 2n/3 \rceil]$ be the graph induced by the central third of $G$.
By  induction on $t$, if $G_{\text{mid}}$ does not contain the pattern $H^-$ with  gap at least $g(n/3, t-1, p)$, then $G_{\text{mid}}$ either satisfies \Pa{} or \Pb\ (note that in the case of \Pa{}, the path may begin at the first vertex of $G_{\text{mid}}$, or end at the last vertex of $G_{\text{mid}}$ since $H^-$ can be either a right or left constellation).
Hence, $G_{\text{mid}}$ either contains an increasing induced path of order at least $f(n/3, t-1, p)$ or an increasing induced path of order at least $h(n/3, t-1, p) \ge f(n/3, t-1, p)$.
By \Cref*{lem:bounds}.\eqref{eq:middlef}, $f(n/3, t-1, p) \ge h(n, t, p)$ which ensures that $G$ indeed satisfies \Pb.

Hence we can assume that $G_{\text{mid}}$ contains the pattern $H^-$ with gap at least $g(n/3, t-1, p)$. 
By \Cref{claim:large-stretch-or-pattern} we know that either $G_{\text{mid}}$ (and $G$) contain the pattern $H$ with gap at least $\frac{g(n/3, t-1, p)-1}{2r + 1}$ or $G_{\text{mid}}$ has stretch at least $\frac{g(n/3, t-1, p)-1}{2r + 1}$.
In the first case, an application of \Cref*{lem:bounds}.\eqref{eq:middleg} gives $\frac{g(n/3, t-1, p)-1}{2r + 1} = s(n, t, p) \geqslant g(n, t, p)$ and thus $G$ satisfies \Pc{}, so we can assume that $G_{\text{mid}}$ has stretch at least $\frac{g(n/3, t-1, p)-1}{2r + 1}$.
Let $G'$ be the successor of $G_{\text{mid}}$.
Notice that 
\begin{equation}
|V(G')| \geqslant \frac{g(n/3, t - 1, p)-1}{2r + 1} = s(n, t, p). \label{eq:gprimes}    
\end{equation}
We now apply induction on $G'$ again with the pattern $H$.

\begin{enumerate}[(P1')]
    \item If $G'$ has an increasing induced path $Q$ with $f(|V(G')|, t, p+1)$ vertices \textbf{starting at its first vertex} (recall that $H$ was assumed to be a right $(t,r)$-constellation),
    then the path $v_1Q$ is an increasing induced path of $G$ starting at $v_1$ and has order at least
    \begin{align*}
        1 + f(|V(G')|, t, p+1) &\geqslant 1 + f(s(n, t, p), t, p+1).&\text{by \eqref{eq:gprimes}}\\
        & \geqslant f(n, t, p) & \text{by \Cref*{lem:bounds}.\eqref{eq:recursionf}},
    \end{align*}
    hence $G$ satisfies \Pa.
    \item If $G'$ has an increasing induced path $Q$ with at least $h(|V(G')|, t, p+1)$ vertices, then observe that
    \begin{align*}
        |Q| &\geqslant h\big(s(n, t, p), t, p+1\big)&\text{by \eqref{eq:gprimes}}\\
        &\geqslant h(n, t, p) &\text{by \Cref*{lem:bounds}.\eqref{eq:recursionh}},
    \end{align*}
    hence $G$ satisfies \Pb.
    \item If $G'$ contains the pattern $H$ with gap at least $g(|V(G')|, t, p+1)$,
    then
    \begin{align*}
        g\big(|V(G')|, t, p+1\big) &\geqslant g\big(s(n, t, p), t, p+1\big)&\text{by \eqref{eq:gprimes}}\\
        &\geqslant g(n, t, p)   &\text{by \Cref*{lem:bounds}.\eqref{eq:recursiong}},
    \end{align*}
    hence $G$ satisfies \Pc.
\end{enumerate}
Hence $\Prop(t, n, p)$ holds for any $n\geq 1$ and $p\ge 0$ (and in particular for $p=0$).
\end{proof}

\subsection{Application to graphs avoiding some topological minor}

We now explain the main application of Theorem \ref{th:peel} on (unordered) graphs avoiding some topological minors. We say that a graph $G$ contains some graph $H$ as a \emph{topological minor} if $G$ contains some subdivision of $H$ as a subgraph (where a \emph{subdivision} of $H$ is a graph obtained from $H$ by replacing each edge by a path of arbitrary length).

\begin{corollary}\label{cor:topomin}
Let $t > 1$ be a positive integer. If a graph $G$ contains an $n$-vertex path $P$ and does not contain $K_t$ as a topological minor, then $G[P]$ contains an induced path of order $(\log n)^{\Omega\left (1/t(\log t)^2\right )}$ which is increasing with respect to $P$.
\end{corollary}
\begin{proof}
Let $G$ be a graph that does not contain $K_t$ as a topological minor and let $P$ be an $n$-vertex path of $G$.
We describe a pattern $H$ that is a $(t,t-1)$-constellation and such that $(G[P],P)$ avoids $H$ as a pattern.

The pattern $H$ consists of $t$ right $(t-1)$-stars $S_1, \dots, S_t$.
For each star $S_i$, we write $c_i$ for its center, and $\ell_{i,1}, \dots, \ell_{i,i-1}, \ell_{i,i+1}, \dots \ell_{i,t}$ for its $t-1$ leaves (note that we omitted the name $\ell_{i,i}$). We then order the vertices of $H$ such that 
\begin{itemize}
\item all centers of the stars lie before all the leaves of the stars, and
\item for any $1\le i<j\le t$, the leaves $\ell_{i,j}$ and $\ell_{j,i}$ are consecutive.
\end{itemize}
Note that there are many orders satisfying these two conditions.
We refer the reader to \Cref{fig:top_minor_pat} for a drawing of such an ordered graph $H$ when $t = 4$.
For each $1\le i<j\le t$, consider the path $P_{i,j}$ which is the concatenation of the edge $c_i\ell_{i,j}$, the subpath of $P$ between $\ell_{i,j}$ and $\ell_{j,i}$, and the edge $\ell_{j,i}c_j$. Note that these paths are internally vertex-disjoint, and thus if $G[P]$ contains the pattern $H$, then it contains $K_t$ as a topological minor.
Hence, $G[P]$ avoids the pattern $H$ and so, by  \Cref{th:peel}, $G[P]$ contains an increasing path of order at least $f(n, t, 0)=(\log n)^{\Omega\left (1/t(\log t)^2\right )}$.
\end{proof}

\begin{figure}[ht]
\centering
\begin{tikzpicture}

\foreach \i in {1, 2, 3, 4} {
    \node (\i) at (\i, 0) {$c_\i$};
}

\foreach \i/\j [count = \c] in {1/2, 2/1, 1/3, 3/1, 1/4, 4/1, 2/3, 3/2, 2/4, 4/2, 3/4, 4/3} {
    \node (\i\j) at (4 + \c, 0)  {$\ell_\i^\j$};
}
\begin{scope}[every node/.style = black node]
\foreach \u/\v in {1/12, 1/13, 1/14} {
	\draw[very thick, blue] (\u) to [bend left=45] (\v);
}
\foreach \u/\v in {2/21, 2/23, 2/24} {
	\draw[very thick, orange] (\u) to [bend left=45] (\v);
}
\foreach \u/\v in {3/31, 3/32, 3/34} {
	\draw[very thick, red] (\u) to [bend left=45] (\v);
}
\definecolor{forestgreen}{RGB}{34, 139, 34}
\foreach \u/\v in {4/41, 4/42, 4/43} {
	\draw[very thick, forestgreen] (\u) to [bend left=45] (\v);
}
\foreach \pos/\name in {1, ..., 4}{
    \fill (\pos + 0.25, 0.25) circle (0.1) {};
}
\foreach \pos/\name in {5, ..., 16}{
    \fill (\pos  - 0.25, 0.25) circle (0.1) {};
}
\end{scope}
\end{tikzpicture}
\caption{Drawing of an ordered graph $H$ with the following property: if a graph $G$ with a Hamiltonian path $P$ contains $H$ as a pattern, then $G$ contains $K_4$ as a topological minor. The ordered graph $H$ consists of four right stars centered respectively in $c_1, c_2, c_3$ and $c_4$, each drawn with a different color.}
\label{fig:top_minor_pat}
\end{figure}
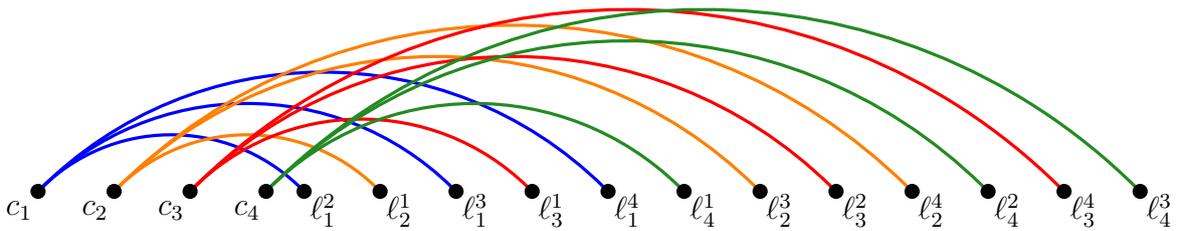

\section{Doubly polylogarithmic upper-bounds}
\label{sec:loglogsqr}

In \cite{defrain2024sparse} Defrain and the third author gave the following upper-bound on the size of induced paths in 2-degenerate graphs with long paths.

\begin{theorem}[\cite{defrain2024sparse}]\label{th:loglogsqr}
There is a constant $c$ such that for infinitely many integers $n$, there is a 2-degenerate  graph $G$ with a path of order $n$ and no induced path of order $c (\log \log n)^2$.
\end{theorem}

In the proof of Theorem~\ref{th:loglogsqr}, the construction of the graph $G$ is explicit.
The graph is obtained after adding subdivisions and extra edges to a 3-blow-up of a binary tree (each node is replaced by a triangle and each edge by two ``parallel'' edges between the triangles corresponding to its endpoints).
We give an alternative construction that implies \Cref{th:loglogsqr} while also having consequences related to excluded patterns.

%\medskip

\begin{restatable}{theorem}{loglogbis}\label{th:loglogbis}
For any ordered graph $H$ which is not a constellation, \[g_H(n) = O((\log \log n)^2).\]
\end{restatable}

We note that our construction is merely a modification of that of \cite{defrain2024sparse} so that it avoids constellations. For completeness we include the proof, but since several parts are similar to the proof of \cite{defrain2024sparse}, we reused as much material from that paper as possible (including definitions, proofs, pictures), with the agreement of the authors. Besides proving \Cref{th:loglogbis}, our contribution here is also to show that the construction of \cite{defrain2024sparse} is quite versatile and in particular, that there is a lot of freedom in choosing the connecting \emph{gadgets} (to be defined below).

\subsection{The construction}

\subsubsection{The base graph.}
Let $h\colon \N_{\geq 1} \to \N$ be the function defined for every $\ell \in \N_{\geq 1}$ by the following formula:
\begin{equation}\label{eq:exp}
h(\ell)=5\cdot 2^{\ell-1} - 2.
\end{equation}
For every $\ell \in \N_{\geq 1}$ we denote by $B_\ell$ the complete binary tree of depth\footnote{The \emph{complete binary tree} of depth $p$ is $K_1$ (rooted at its unique vertex) if $p=1$ and otherwise it can be obtained from two disjoint copies of the complete binary tree of depth $p-1$ by adding a new vertex $v$ adjacent to their roots and rooting the resulting tree at $v$.} $h(\ell)$. In this tree, $\prec$ is the ancestor-descendant relation, i.e.\ $s\prec t$ if $s\neq t$ and $s$ lies on the unique path linking $t$ to the root.
Let $H_\ell$ be the graph obtained from $B_\ell$ by replacing each vertex by the subgraph called \emph{gadget} that is drawn in Figure~\ref{fig:gadget} (while Figure~\ref{fig:constr}, right, shows the gadgets used in \cite{defrain2024sparse} as a comparison) and connected to other gadgets as described in Figure~\ref{fig:constr}, bottom (and explained more formally below, after we introduce some necessary terminology). So, each vertex $s$ of $B_\ell$ corresponds to a different gadget in $H_\ell$, that we refer to as the \emph{gadget at $s$}. Conversely, it is useful for our proofs to define a function $\pi$ that maps the vertices of $H_\ell$ back to the vertex of $B_\ell$ they originate from. Hence, the gadget at $s$ is precisely the subgraph of $H_\ell$ induced by $\pi^{-1}(s)$. We define the depth $\depth(u)$ of a vertex $u\in V(H_\ell)$ as the depth of its corresponding node $\pi(u)$ in $B_\ell$. Gadgets have three special sets of vertices as described on Figure~\ref{fig:gadget}, the \emph{in-ports}, the \emph{out-ports}, and the \emph{connectors}. %For $s\in V(B_\ell)$, we denote by $\pred(s)$ the in-ports of the gadget at $s$ in $H_\ell$. 

For any non-leaf node $s$ of $B_\ell$, say with left child $s_1$ and right child $s_2$, we add 
a matching between the left connectors of the gadget at $s$ and the out-ports of the gadget at $s_1$, and another matching between the right connectors of the gadget at $s$ and the out-ports of the gadget at $s_2$. 
The resulting graph is $H_\ell$. This is illustrated in Figure~\ref{fig:constr}, bottom.
In the following we will also add edges to $H_\ell$ between in-ports and out-ports of specified depths.

\begin{figure}[htb]
    \centering
    \begin{tikzpicture}[every node/.style={draw, circle, fill=black, minimum size=0.15cm, inner sep=0}]
    \draw (-90:1) node[draw=none, fill=none] (S) {} (30:1) node (NE) {} (150:1) node (NW) {};
    \draw (NE) -- ++(0, 1) node[star, minimum size=2.75mm, color=orange] (NE1) {};
    \draw (NW) -- ++(0, 1) node[star,minimum size=2.75mm, color=orange] (NW1) {};
    \draw (NW1) -- (NE);
    \draw (NE1) -- (NW);
    \node [draw = none, fill = none] () [above right of = NE1] {right out-port};
    \node [draw = none, fill = none] () [above left of = NW1] {left out-port};
    
    \draw (S) ++(-30:0.5) node[regular polygon, regular polygon sides=5, minimum size=2.75mm,color=violet] (SSE) {} -- ++(-30:0.5) node[regular polygon, regular polygon sides=5, minimum size=2.75mm, color=violet] (SSE1) {} -- ++(-30:0.5) node[color=green!50!black, minimum size=2.5mm] (SSE2) {};
    \draw (NE) -- ++(-30:0.5) node[regular polygon, regular polygon sides=5, minimum size=2.75mm, color=violet] (NESE) {} -- ++(-30:0.5) node[regular polygon, regular polygon sides=5, minimum size=2.75mm, color=violet] (NESE1) {} -- ++(-30:0.5) node[color=green!50!black, minimum size=2.5mm] (NESE2) {};
    \draw (SSE2) -- (NESE2);
    \draw (S) ++(210:0.5) node[regular polygon, regular polygon sides=5, minimum size=2.75mm, color=violet] (SSW) {} -- ++(210:0.5) node[regular polygon, regular polygon sides=5, minimum size=2.75mm, color=violet] (SSW1) {} -- ++(210:0.5) node[color=green!50!black, minimum size=2.5mm] (SSW2) {};
    \draw (NW) -- ++(210:0.5) node[regular polygon, regular polygon sides=5, minimum size=2.75mm, color=violet] (NWSW) {} -- ++(210:0.5) node[regular polygon, regular polygon sides=5, minimum size=2.75mm, color=violet] (NWSW1) {} -- ++(210:0.5) node[color=green!50!black, minimum size=2.5mm] (NWSW2) {};
    \draw (SSW2) -- (NWSW2);
    \draw (SSE) -- (SSW);
    \end{tikzpicture}
    \caption{Gadget used to replace vertices of $B_\ell$. The \emph{out-ports} (resp.\ \emph{in-ports}, \emph{connectors}) are depicted with colored stars (resp.\ pentagons, circles).}
    \label{fig:gadget}
\end{figure}
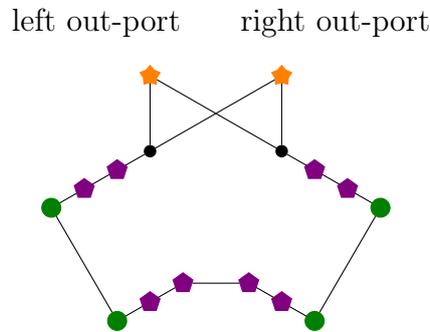

\begin{figure}[htb]
    \centering
    \begin{tikzpicture}[every node/.style={draw, circle, fill, minimum size=0.15cm, inner sep=0}, every path/.style={thick}]
    \begin{scope}
    \draw (0,0) node[color=Purple] (Z) {};        \draw[dashed] (Z) -- ++(0, 1);
    \draw (-30:3.5) node[color=orange] (SE) {} -- (Z) -- (210:3.5) node[color=BlueGreen] (SW) {};
    \draw[dashed] (SE) -- ++(30:1);
    \draw[dashed] (SE) -- ++(-90:1);
    \draw[dashed] (SW) -- ++(-210:1);
    \draw[dashed] (SW) -- ++(-90:1);
    \end{scope}
    \begin{scope}[xshift = 0cm, yshift = -9cm]
    %% Construction avec notre gadget
    \begin{scope}[color = Purple, scale=2/3]
    \draw (-90:1) node[draw=none, fill=none] (Sa) {} (30:1) node (NEa) {} (150:1) node (NWa) {};
    \draw (NEa) -- ++(0, 1) node (NE1a) {};
    \draw (NWa) -- ++(0, 1) node (NW1a) {};
    \draw (NW1a) -- (NEa);
    \draw (NE1a) -- (NWa);
    
    \draw (Sa) ++(-30:0.5) node (SSEa) {} -- ++(-30:0.5) node (SSE1a) {} -- ++(-30:0.5) node (SSE2a) {};
    \draw (NEa) -- ++(-30:0.5) node (NESEa) {} -- ++(-30:0.5) node (NESE1a) {} -- ++(-30:0.5) node (NESE2a) {};
    \draw (SSE2a) -- (NESE2a);
    \draw (Sa) ++(210:0.5) node (SSWa) {} -- ++(210:0.5) node (SSW1a) {} -- ++(210:0.5) node (SSW2a) {};
    \draw (NWa) -- ++(210:0.5) node (NWSWa) {} -- ++(210:0.5) node (NWSW1a) {} -- ++(210:0.5) node (NWSW2a) {};
    \draw (SSW2a) -- (NWSW2a);
    \draw (SSEa) -- (SSWa);
    \end{scope}
    \begin{scope}[{shift=(210:3.25)}, scale = 2/3, rotate = -60, color = BlueGreen, fill = BlueGreen]
        % Nw
    \draw (-90:1) node[draw=none, fill=none] (Sb) {} (30:1) node (NEb) {} (150:1) node (NWb) {};
    \draw (NEb) -- ++(0, 1) node (NE1b) {};
    \draw (NWb) -- ++(0, 1) node (NW1b) {};
    \draw (NW1b) -- (NEb);
    \draw (NE1b) -- (NWb);
    
    \draw (Sb) ++(-30:0.5) node (SSEb) {} -- ++(-30:0.5) node (SSE1b) {} -- ++(-30:0.5) node (SSE2b) {};
    \draw (NEb) -- ++(-30:0.5) node (NESEb) {} -- ++(-30:0.5) node (NESE1b) {} -- ++(-30:0.5) node (NESE2b) {};
    \draw (SSE2b) -- (NESE2b);
    \draw (Sb) ++(210:0.5) node (SSWb) {} -- ++(210:0.5) node (SSW1b) {} -- ++(210:0.5) node (SSW2b) {};
    \draw (NWb) -- ++(210:0.5) node (NWSWb) {} -- ++(210:0.5) node (NWSW1b) {} -- ++(210:0.5) node (NWSW2b) {};
    \draw (SSW2b) -- (NWSW2b);
    \draw (SSEb) -- (SSWb);
    \end{scope}
    \begin{scope}[{shift=(-30:3.25)}, scale = 2/3, rotate=60, color = orange]
        % NE
    \draw (-90:1) node[draw=none, fill=none] (Sc) {} (30:1) node (NEc) {} (150:1) node (NWc) {};
    \draw (NEc) -- ++(0, 1) node (NE1c) {};
    \draw (NWc) -- ++(0, 1) node (NW1c) {};
    \draw (NW1c) -- (NEc);
    \draw (NE1c) -- (NWc);
    
    \draw (Sc) ++(-30:0.5) node (SSEc) {} -- ++(-30:0.5) node (SSE1c) {} -- ++(-30:0.5) node (SSE2c) {};
    \draw (NEc) -- ++(-30:0.5) node (NESEc) {} -- ++(-30:0.5) node (NESE1c) {} -- ++(-30:0.5) node (NESE2c) {};
    \draw (SSE2c) -- (NESE2c);
    \draw (Sc) ++(210:0.5) node (SSWc) {} -- ++(210:0.5) node (SSW1c) {} -- ++(210:0.5) node (SSW2c) {};
    \draw (NWc) -- ++(210:0.5) node (NWSWc) {} -- ++(210:0.5) node (NWSW1c) {} -- ++(210:0.5) node (NWSW2c) {};
    \draw (SSW2c) -- (NWSW2c);
    \draw (SSEc) -- (SSWc);
    \end{scope}
    \draw (SSW2a) -- (NE1b) (NWSW2a) -- (NW1b);
    \draw (SSE2a) -- (NW1c) (NESE2a) -- (NE1c);
    \draw[dashed] (SSE2b) -- ++(-90:0.5) (NESE2b) -- ++(-90:0.5) (SSW2c) -- ++(-90:0.5) (NWSW2c) -- ++(-90:0.5); 
    \draw[dashed] (NWSW2b) -- ++(150:0.5) (SSW2b) -- ++(150:0.5);
    \draw[dashed] (NESE2c) -- ++(30:0.5) (SSE2c) -- ++(30:0.5);
    \draw[dashed] (NE1a) --++(0,0.5) (NW1a) --++(0,0.5);
    \end{scope}
    % Construction originale:
    \begin{scope}[xshift = 6.5cm, yshift = -6cm, scale=2/3]
    \begin{scope}[color=Purple]
    \draw (-90:1) node (aS) {} -- (30:1) node (aNE) {} -- (150:1) node (aNW) {} --(aS);
    \draw (aNW) --++(0,1) node (aNW1) {} node[midway] {};
    \draw (aNE) --++(0,1) node (aNE1) {} node[midway] {};
    \end{scope}
    \begin{scope}[{shift=(210:3)}, rotate = -60, color = BlueGreen]
    \draw (-90:1) node (bS) {} -- (30:1) node (bNE) {} -- (150:1) node (bNW) {} --(bS);
    \draw (bNW) --++(0,1) node (bNW1) {} node[midway] {};
    \draw (bNE) --++(0,1) node (bNE1) {} node[midway] {};
    \end{scope}
    \begin{scope}[{shift=(-30:3)}, rotate = 60, color = orange]
    \draw (-90:1) node (cS) {} -- (30:1) node (cNE) {} -- (150:1) node (cNW) {} --(cS);
    \draw (cNW) --++(0,1) node (cNW1) {} node[midway] {};
    \draw (cNE) --++(0,1) node (cNE1) {} node[midway] {};
    \end{scope}
    \draw (bNW1) -- (aNW) (bNE1) -- (aS) (cNW1) -- (aS) (cNE1) -- (aNE);
    \draw[dashed] (aNE1) --++(0,0.5) (aNW1) --++(0,0.5);
    \draw[dashed] (cNW) -- ++(0,-0.5) (cS) -- ++(0,-0.5) (bNE) -- ++(0, -0.5) (bS) -- ++(0, -0.5);
    \draw[dashed] (cS) --++(30:0.5) (cNE) --++(30:0.5) (bS) --++(150:0.5) (bNW) --++(150:0.5);
    \end{scope}
    \draw[very thick, -latex] (0, -4) --++(-0,-1.5);
    \draw[very thick, -latex] (3, -4) --++(-45:1.5);
    \end{tikzpicture}
    \caption{Our modification (bottom) of the definition of \cite{defrain2024sparse} (right) of $H_\ell$ from a binary tree.}
    \label{fig:constr}
\end{figure}
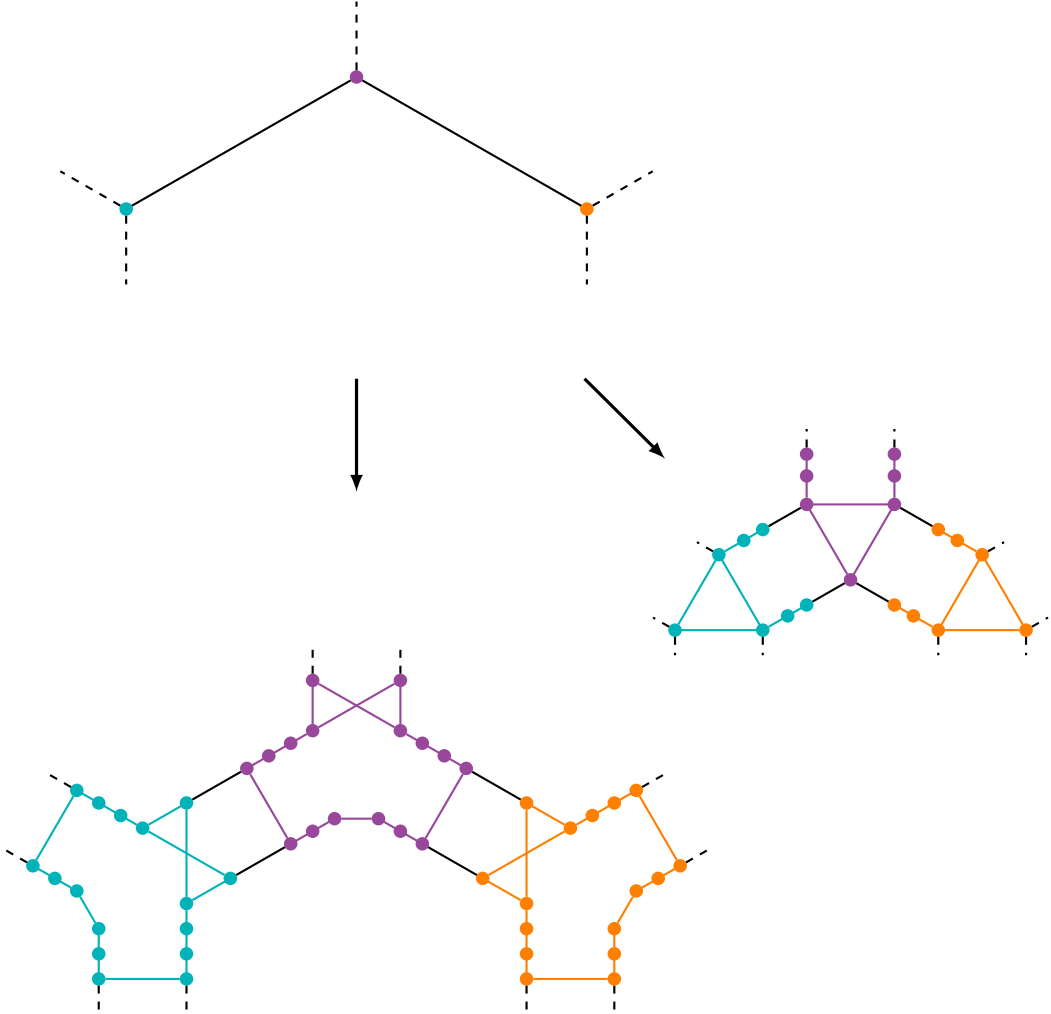

\subsubsection{Nested intervals systems.}

If $X$ is a set of pairs of integers and $i \in \N$, we denote by $X\!\oplus i$ the set $\{(x+i, x'+i) : (x,x') \in X\}$.
For every $\ell \in \N_{\geq 1}$ the set $\cN_\ell$ is recursively defined as follows:
\[
\left \{
\begin{array}{ll}
     \cN_1 = \{(1,3)\},\ \text{and}&  \\
     \cN_\ell = \{(1, h(\ell))\}\ \cup\ (\cN_{\ell-1}\oplus 1)\ \cup\ \big(\cN_{\ell-1} \oplus (h(\ell-1) +1)\big)& \text{if}\ \ell>1.
\end{array}
\right.
\]
The elements of $\cN_\ell$ are called \emph{intervals}.
Intuitively, $\cN_\ell$ is obtained by taking two copies of $\cN_{\ell-1}$ (appropriately shifted so that they start after integer 1 and do not intersect) and adding a new interval $(1,h(\ell))$ containing the two copies. See Figure~\ref{fig:intervals} for an illustration.

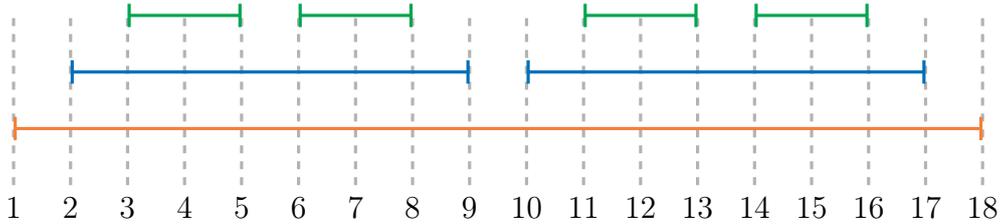
\begin{figure}
    \centering
    \begin{tikzpicture}[every path/.style = {very thick}, scale=0.75]
    \foreach \i in {1,..., 18}{
        \draw[dashed, color = black!30] (\i, 0) node[anchor=north, color=black] (n\i) {\i} -- ++(0,3);
    }
    \draw[|-|,color=Orange] (1,1) -- (18,1);
    \draw[|-|, color=NavyBlue] (2,2) -- (9,2);
    \draw[|-|, color=NavyBlue] (10,2) -- (17,2);
    \draw[|-|, color = Green] (3,3) -- (5,3);
    \draw[|-|, color = Green] (6,3) -- (8,3);
    \draw[|-|, color = Green] (11,3) -- (13,3);
    \draw[|-|, color = Green] (14,3) -- (16,3);
    \end{tikzpicture}
    \caption{The intervals of $\cN_3$ with intervals of rank $1,\ 2$, and $3$ depicted from top to bottom in green, blue, and orange, respectively.}\label{fig:intervals}
\end{figure}
The following easy properties of $\cN_{\ell}$ can be proved by a straightforward induction:

\begin{remark}\label{rem:nestint}
For every $\ell\in \N_{\geq 1}$ the following holds:
\begin{enumerate}
    \item the endpoints of the intervals in $\cN_{\ell}$ range from $1$ to $h(\ell)$ and are all distinct;
    \item every interval of $\cN_\ell$ is of the form $(i, i+h(a)-1)$ for some $i \in \intv{1}{h(\ell)}$ and $a \in \intv{1}{\ell}$;
    \item for every interval $(i,j)$ in $\cN_\ell$ there is no other interval $(i', j')$ in $\cN_\ell$ such that $i<i'<j<j'$ (informally, intervals do not cross).
\end{enumerate}
\end{remark}

We call \emph{rank} of an interval $(i,j)\in \cN_\ell$ the aforementioned integer $a \in \intv{1}{\ell}$ such that $j = i+h(a)-1$.

\subsubsection{Ribs.}
The function $\ribs_\ell$ is defined on every pair $(s,t) \in V(B_\ell)^2$ of nodes such that $s \prec t$ as the set of edges between $\pi^{-1}(s)$ and $\pi^{-1}(t)$ described in Figure~\ref{fig:ribs}.\footnote{Note that these edges do not exist in $H_\ell$. We define this set in order to later construct a graph by adding these edges to $H_\ell$.}

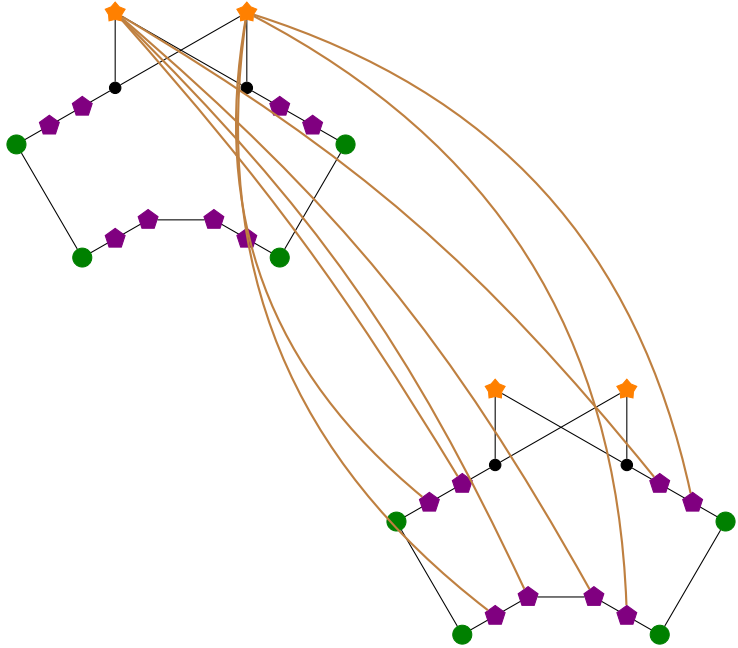
\begin{figure}[htb]
    \centering
    \begin{tikzpicture}[every node/.style=black node]
    
    \begin{scope}
    \draw (-90:1) node[draw=none, fill=none] (Sa) {} (30:1) node (NEa) {} (150:1) node (NWa) {};
    \draw (NEa) -- ++(0, 1) node[star, minimum size=2.75mm, color=orange] (NE1a) {};
    \draw (NWa) -- ++(0, 1) node[star,minimum size=2.75mm, color=orange] (NW1a) {};
    \draw (NW1a) -- (NEa);
    \draw (NE1a) -- (NWa);
    
    \draw (Sa) ++(-30:0.5) node[regular polygon, regular polygon sides=5, minimum size=2.75mm,color=violet] (SSEa) {} -- ++(-30:0.5) node[regular polygon, regular polygon sides=5, minimum size=2.75mm, color=violet] (SSE1a) {} -- ++(-30:0.5) node[color=green!50!black, minimum size=2.5mm] (SSE2a) {};
    \draw (NEa) -- ++(-30:0.5) node[regular polygon, regular polygon sides=5, minimum size=2.75mm, color=violet] (NESEa) {} -- ++(-30:0.5) node[regular polygon, regular polygon sides=5, minimum size=2.75mm, color=violet] (NESE1a) {} -- ++(-30:0.5) node[color=green!50!black, minimum size=2.5mm] (NESE2a) {};
    \draw (SSE2a) -- (NESE2a);
    \draw (Sa) ++(210:0.5) node[regular polygon, regular polygon sides=5, minimum size=2.75mm, color=violet] (SSWa) {} -- ++(210:0.5) node[regular polygon, regular polygon sides=5, minimum size=2.75mm, color=violet] (SSW1a) {} -- ++(210:0.5) node[color=green!50!black, minimum size=2.5mm] (SSW2a) {};
    \draw (NWa) -- ++(210:0.5) node[regular polygon, regular polygon sides=5, minimum size=2.75mm, color=violet] (NWSWa) {} -- ++(210:0.5) node[regular polygon, regular polygon sides=5, minimum size=2.75mm, color=violet] (NWSW1a) {} -- ++(210:0.5) node[color=green!50!black, minimum size=2.5mm] (NWSW2a) {};
    \draw (SSW2a) -- (NWSW2a);
    \draw (SSEa) -- (SSWa);
    \end{scope}
    \begin{scope}[shift = {(5,-5)}]
    \draw (-90:1) node[draw=none, fill=none] (Sb) {} (30:1) node (NEb) {} (150:1) node (NWb) {};
    \draw (NEb) -- ++(0, 1) node[star, minimum size=2.75mm, color=orange] (NE1b) {};
    \draw (NWb) -- ++(0, 1) node[star,minimum size=2.75mm, color=orange] (NW1b) {};
    \draw (NW1b) -- (NEb);
    \draw (NE1b) -- (NWb);
    
    \draw (Sb) ++(-30:0.5) node[regular polygon, regular polygon sides=5, minimum size=2.75mm,color=violet] (SSEb) {} -- ++(-30:0.5) node[regular polygon, regular polygon sides=5, minimum size=2.75mm, color=violet] (SSE1b) {} -- ++(-30:0.5) node[color=green!50!black, minimum size=2.5mm] (SSE2b) {};
    \draw (NEb) -- ++(-30:0.5) node[regular polygon, regular polygon sides=5, minimum size=2.75mm, color=violet] (NESEb) {} -- ++(-30:0.5) node[regular polygon, regular polygon sides=5, minimum size=2.75mm, color=violet] (NESE1b) {} -- ++(-30:0.5) node[color=green!50!black, minimum size=2.5mm] (NESE2b) {};
    \draw (SSE2b) -- (NESE2b);
    \draw (Sb) ++(210:0.5) node[regular polygon, regular polygon sides=5, minimum size=2.75mm, color=violet] (SSWb) {} -- ++(210:0.5) node[regular polygon, regular polygon sides=5, minimum size=2.75mm, color=violet] (SSW1b) {} -- ++(210:0.5) node[color=green!50!black, minimum size=2.5mm] (SSW2b) {};
    \draw (NWb) -- ++(210:0.5) node[regular polygon, regular polygon sides=5, minimum size=2.75mm, color=violet] (NWSWb) {} -- ++(210:0.5) node[regular polygon, regular polygon sides=5, minimum size=2.75mm, color=violet] (NWSW1b) {} -- ++(210:0.5) node[color=green!50!black, minimum size=2.5mm] (NWSW2b) {};
    \draw (SSW2b) -- (NWSW2b);
    \draw (SSEb) -- (SSWb);
    \end{scope}
    \draw[thick, color=brown]
    (NE1a) to[bend left=30] (NESE1b)
    (NW1a) to[bend left=10] (NESEb);
    \draw[thick, color=brown]
    (NE1a) to[bend right=30] (NWSW1b)
    (NW1a) to[bend left=5] (NWSWb);
    \draw[thick, color=brown]
    (NE1a) to[bend left=30] (SSE1b)
    (NW1a) to[bend left=10] (SSEb);
    \draw[thick, color=brown]
    (NE1a) to[bend right=30] (SSW1b)
    (NW1a) to[bend left=10] (SSWb);
        
    \end{tikzpicture}
    \caption{In color, the set $\ribs(s,t)$ of extra edges from out-ports of $\pi^{-1}(s)$ (top left) to in-ports of $\pi^{-1}(t)$ (bottom right).}
    \label{fig:ribs}
\end{figure}

The graph $G_\ell$ is obtained from $H_\ell$ after the addition of the set of edges $\ribs_\ell(s,t)$ for every pair of nodes $s,t\in V(B_\ell)$ of respective depth $i,j$ such that $s\prec t$ and $(i,j)\in \cN_\ell$.
We call an edge $uv$ in that set a \emph{rib}.
Hence the edges of $G_\ell$ are partitioned into ribs, \emph{gadget edges} (i.e., edges with both endpoints in the same gadget) and \emph{tree edges} (edges connecting two gadgets corresponding to adjacent nodes of $B_\ell$).

The following stems from the definition of a gadget.
\begin{remark}\label{rem:preimage-partition}
    The family of sets $\{\pi^{-1}(t) : t\in V(B_\ell)\}$ defines a partition of $V(G_\ell)$ with $|\pi^{-1}(t)|=16$ for every node $t\in V(B_\ell)$.
\end{remark}

\subsection{The properties of \texorpdfstring{$G_\ell$}{Gl}}

In this section, we describe the properties of the construction.
The proofs are straightforward and very similar to those in \cite{defrain2024sparse}, so we omit them and refer to the proofs of the corresponding results in \cite{defrain2024sparse}.

\begin{lemma}\label{lem:ham}
    For every integer $\ell \geq 1$, the graph $G_\ell$ has a Hamiltonian path.
\end{lemma}

\begin{lemma}\label{lem:taillede}
For every integer $\ell\geq 1$,    $|V(G_\ell)| \geq 2^{2^{\ell}}$.
\end{lemma}
\begin{proof}
Indeed, $B_\ell$ is a complete binary tree of depth $h(\ell)$. Hence, $B_\ell$ has $2^{h(\ell)} - 1$ nodes. Since each node of $B_\ell$ is replaced by a copy of the 16-vertex gadget to make $G_\ell$, we have $|V(G_\ell)| = 16(2^{h(\ell)} - 1)  \geq 2^{5 \cdot 2^{\ell - 1} - 2} - 1 \ge 2^{2^{\ell} + 2^{\ell-1}} - 1 \ge 2^{2^\ell}$.
\end{proof}

\begin{lemma}\label{lem:2deg}
    For every integer $\ell\geq 1$,
    the graph $G_\ell$ is 2-degenerate.
\end{lemma}

Let us note that the proofs in the upcoming  \Cref{sec:ribs} do not depend on gadgets themselves but on the structure of the ribs. In the current proof, the properties of the gadgets are only used to show the above lemmas and in the last step of \Cref{sec:smallip}.

\subsection{Ribs, sources, and their properties}
\label{sec:ribs}

In the rest of the proof we fix $\ell\in\mathbb{N}_{\geq 1}$.
A node $s$ of the tree $B_\ell$ is a \emph{source} if there is an interval $(i,j) \in \cN_\ell$ such that $s$ has depth~$i$.
Intuitively, this means that in $G_\ell$ there are ribs from the out-ports of the gadget at $s$ to the in-ports of the gadget at $t$, for every descendant $t$ of $s$ of depth~$j$.
For a source $s$ of $B_\ell$ the \emph{rank} of $s$ is defined as the rank of $(i,j)$, i.e., the integer $a\in\intv{1}{\ell}$ such that $j - i + 1 = h(a)$. As for depth, we will extend this notation to gadgets and vertices of $G_\ell$: $\rank(J) := \rank(s)$ if $J$ is the gadget at $s$, and $\rank(v) := \rank(J)$ if $v \in V(J)$.

We denote by $B_\ell(s)$ the subtree of $B_\ell$ rooted at $s$ and of depth~$h(a)$. This means that the leaves of $B_\ell(s)$ are exactly those vertices $t$ such that in $G_\ell$, the gadget at $s$ sends ribs to the gadget at $t$.
The graph $G_\ell(s)$ is defined as the subgraph of $G_\ell$ induced by $\pi^{-1}\big(V(B_\ell(s))\big)$.

The \emph{internal nodes} of $B_\ell(s)$ are those that are neither the root or leaves of $B_\ell(s)$.
For every node $x$ of $B_\ell$, we define $\tau(x)$ as the minimum rank of a source $s$ such that $x$ is an internal node of $B_\ell(s)$. Notice that if $x$ is the root or a leaf of $B_\ell$ then $\tau$ is not defined: in this case we set $\tau(x) = \ell+1$.
We naturally extend the definition of $\tau$  to gadgets and vertices of $G_\ell$ as we did for $\rank$ above.

\smallskip

In a graph $G$, we say that a set $X\subseteq V(G)$ \emph{separates} two sets $Y,Z\subseteq V(G)$ if every path from a vertex of $Y$ to a vertex of $Z$ intersects $X$.

\begin{remark}\label{rem:edge-order}
If two vertices $u$ and $v$ are adjacent in $G_\ell$, then $\pi(u)$ and $\pi(v)$ are $\preceq$-comparable. In particular, if $\depth(u) = \depth(v)$ then $\pi(u) = \pi(v)$.
\end{remark}

\begin{remark}\label{rem:source-sep}
If $s$ is a source of $B_\ell$, and $\mathcal{L}$ is the set of leaves of $B_\ell(s)$, then
$$ X :=  \bigcup\limits_{x \in \mathcal{L}\cup \{s\}} \pi^{-1}(x) $$ 
separates $G_\ell(s)$ from $G_\ell \setminus G_\ell(s)$.
\end{remark}
\begin{remark}\label{rem:intertau}
    Let $s$ be a source of $B_\ell$ of rank $a\in \intv{1}{\ell}$ and $v \in V(G_\ell(s))$.
    \begin{enumerate}
        \item if $\pi(v)$ is an internal node of $B_\ell(s)$ then $\tau(v)\leq a$;
        \item if $\pi(v) = s$ or $\pi(v)$ is a leaf of $B_\ell(s)$, then $\tau(v) = a+1$.
    \end{enumerate}
\end{remark}

\begin{remark}\label{rem:depth-decreasing}
If $uv$ is an edge of $G_\ell$ such that $\depth(u)< \depth(v)$, then either $uv$ is a tree edge, or $uv$ is a rib of source $\pi(u)$. In the first case, $u$ is a connector and $v$ an out-port, and in the second case, $u$ is an out-port and $v$ an in-port.
\end{remark}

\begin{lemma}\label{lem:tau-decreasing}
If $uv$ is an edge of $G_\ell$ such that $\tau(u) > \tau(v)$, then $uv$ is a tree edge and $\tau(u)=\tau(v)+1$.
If in addition, $\depth(v)\geq \depth(u)$, then $\pi(u)$ is a source,
$u$ is a connector, and $v$ is an in-port.
%$u \in V(K^{\pi(u)})$.
\end{lemma}

\begin{proof}
Since $\tau(v) < \tau(u)$ we have $\tau(v) \le \ell$.
Hence there exists a source $s$ of minimum rank such that $\pi(v)$ is an internal node of $B_\ell(s)$.
By definition we have $\rank(s) = \tau(v)$.

By \Cref{rem:source-sep}, the vertex $u$ being adjacent to $v$, we obtain $u \in V(G_\ell(s))$.
Since $\tau(u) > \rank(s)$, \Cref{rem:intertau} asserts that $\pi(u)$ is either $s$ or a leaf of $B_\ell(s)$.
In particular, $\tau(u) = \rank(s) + 1 = \tau(v) + 1$.

Finally, since $u$ and $v$ are not in the same gadget, \Cref{rem:edge-order} asserts that $\depth(u) \neq \depth(v)$.
Hence by \Cref{rem:depth-decreasing} $uv$ is either a tree edge or a rib; but it cannot be a rib as $\tau(u) \neq \tau(v)$. Hence, $uv$ is a tree edge.
As the depth of the leaves of $B_\ell(s)$ is larger than the depth of its internal nodes, the case where $\depth(v) \ge \depth(u)$ translates into $\pi(u) = s$, $u$ is a connector and $v$ an in-port.
\end{proof}

\newcommand{\Out}{\mathrm{Out}}
\newcommand{\Int}{\mathrm{Int}}
We will need a last structural lemma, which gives a slightly more precise version of \Cref{rem:source-sep}.
For any source $s$, we denote the outside of $G_\ell(s)$ by $\Out(s) := G_\ell \setminus G_\ell(s)$ and its interior by $\Int(s) := \{v \in V(G_\ell) \ : \ \pi(v) \text{ is an internal node of } B_\ell(s)\}$. Recall that for a vertex $v$ in a graph $G$, $N_G[v]=\{v\}\cup N_G(v)$ denotes the closed neighborhood of $v$ in $G$ (we omit subscripts when $G$ is clear from the context).

\begin{lemma}\label{lem:separation}
Let $s$ be a source of $B_\ell$. Let $u_L$ (resp.\ $u_R$) be the left (resp.\ right) out-port of the gadget at $s$. Then,
\begin{itemize}
    \item \label{it:left_sep} the set $N_{G_\ell(s)}[u_L]$ separates $\Out(s) \cup \{u_R\}$ from $\Int(s)$;~and
    \item \label{it:right_sep} the set $N_{G_\ell(s)}[u_R] \cup \{u_L\}$ separates $\Out(s)$ from $\Int(s)$.
\end{itemize}
\end{lemma}

\begin{proof}
Let $P = v_1, \dots, v_p$ be any path from $\Int(s)$ to $\Out(s)$.
Consider the smallest $i$ such that $v_{i+1} \in \Out(s)$ and $v_i \not \in \Out(s)$.

Note that $v_i$ and $v_{i+1}$ cannot be in the same gadget. Hence, $\depth(v_i) \neq \depth(v_{i+1})$.
As $s$ is a source, there is an interval $(a,b) \in \cN_\ell$ such that $s$ has depth $a$. The intervals of $\cN_\ell$ do not cross and have distinct endpoints (\Cref{rem:nestint}), so by definition of ribs the edge $v_iv_{i+1}$ is not a rib. Hence, $v_iv_{i+1}$ is a tree edge.
Tree edges connect gadgets at adjacent nodes of $B_\ell$. So, either $v_i$ is an out-port of the gadget at $s$ and $v_{i+1}$ is a connector of the parent of $s$, or $v_i$ is a connector of a gadget at some leaf of $B_\ell(s)$ and $v_{i+1}$ is an out-port of the gadget at a child of this leaf.

We denote by $C$ the set of connector vertices in the gadgets at leaves of $B_\ell(s)$.
By construction, any neighbor of a vertex in $C$ in $G_\ell(s)$ is an in-port adjacent to $u_R$, i.e., $N_{G_\ell(s)}(C) \cap V(G_\ell(s)) \subseteq N_{G_\ell(s)}(u_R)$. Similarly, $N_{G_\ell(s)}(N_{G_\ell(s)}(C)) \cap V(G_\ell(s)) \subseteq N_{G_\ell(s)}(u_L)$.
Hence any path from $\Int(s)$ to $C$ that stays in $G_\ell(s)$ will intersect both $N_{G_\ell(s)}(u_R)$ and $N_{G_\ell(s)}(u_L)$. Hence, a path going from $\Int(s)$ to $\Out(s)$ intersects either both $N_{G_\ell(s)}(u_L)$ and $N_{G_\ell(s)}(u_R)$, or one of $\{u_L, u_R\}$. If $v_i \in \{u_L, u_R\}$ then either $v_i = u_L$, or $v_i = u_R$ in which case $v_{i-1} \in N(u_R)$, which implies that the subpath $v_1 ,\dots, v_{i-1}$ intersects $N_{G_\ell(s)}[u_L]$.
\end{proof}

\subsection{Special sources and length of induced paths}

In this section, let $Q$ be an induced path of $G_\ell$.

\begin{lemma}\label{lem:depth-root}
There is a unique node $t \in V(B_\ell)$ of minimum depth subject to $\pi^{-1}(t) \cap V(Q)\neq \emptyset$.
\end{lemma}

\begin{proof}
Let us assume towards a contradiction that there are two different such nodes $t,t'$. As they have the same depth, they are not $\preceq$-comparable in $B_\ell$. Recall that every edge of $G_\ell$ connects vertices whose image by $\pi$ is $\preceq$-comparable. Therefore, the subpath of $Q$ linking $\pi^{-1}(t)$ to $\pi^{-1}(t')$ contains a vertex of $\pi^{-1}(t'')$, for some common ancestor $t''$ of $t$ and $t'$, which contradicts the minimality of the depth of those vertices.
\end{proof}

\begin{lemma}\label{lem:tau-constant}
There is a constant $c_{\ref{lem:tau-constant}}$ such that if $Q=u_1,\dots ,u_q$ is an induced path of $G_\ell$ with $\tau(u_i)=\tau(u_1)$ for all $2\leq i\leq q$, then $|Q|\leq c_{\ref{lem:tau-constant}}$.
\end{lemma}

\begin{proof}
Let $Q$ be such an induced path and $s$ be the source of rank $a=\tau(u_1)$ such that $\pi(u_1)$ is an internal node of $B_\ell(s)$.
In order to show the statement of the lemma, we will prove that $Q$ visits a bounded number of distinct gadgets. This is enough since gadgets have bounded size (Remark~\ref{rem:preimage-partition}).

By Remark~\ref{rem:intertau}, if a vertex $v\in V(G_\ell)$ belongs to the gadget at $s$ or at some leaf of $B_\ell(s)$ then $\tau(v) = a+1$. By Lemma~\ref{lem:separation}, this implies that $Q$ is contained in the union of the gadgets at the internal nodes of $B_\ell(s)$. Let us call $Z$ the union of the vertex sets of these gadgets.

If $a=1$ there are at most two such gadgets so we are done. So, we may now assume $a>1$.

As $s$ is a source, there is an interval $(i,i'') \in \cN_\ell$ such that $\depth(s) = i$ and, for every leaf $t$ of $B_\ell(s)$, $\depth(t) = i''$. We call $s_1$ and $s_2$ the two children of~$s$. Let $i' = i+ h(a-1)$. By construction, $(i+1, i'), (i'+1, i''-1) \in \cN_\ell$; see Figure~\ref{fig:onigiri} for a representation of $B_\ell(s)$. (Notice that $i''\geq i+2$, by the third item of Remark~\ref{rem:nestint}, the definition of the function $h$ and the fact that $a>1$.)
Let $D$ be the set of descendants of $s$ that have depth $i'+1$ in $G_\ell$ (the colored nodes in the figure).

\begin{figure}[htb]
    \centering
    \includegraphics[scale=1.1]{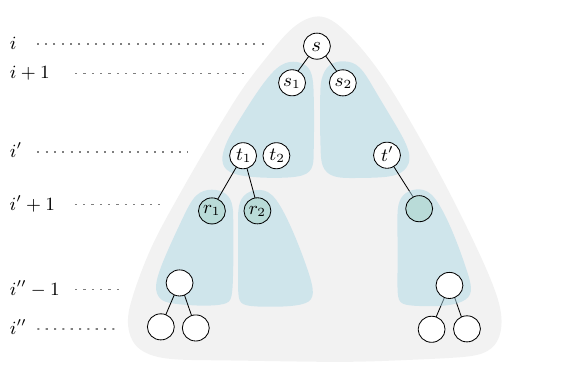}
    \caption{The situation in Lemma~\ref{lem:tau-constant}. Picture from \cite{defrain2024sparse}.}\label{fig:onigiri}
\end{figure}

Let $r \in D$ and let $t$ be a leaf of $B_\ell(r)$. Then, $t$ lies at depth $i''-1$ in $B_\ell$. Observe that in $G_\ell$ each edge with only one endpoint in the gadget at $t$ is of one of the following types:

\begin{itemize}
    \item a tree edge from a vertex of $\pi^{-1}(t^*)$ where $t^*$ is the parent or a child of $t$; or
    \item a rib from an out-port of $\pi^{-1}(r)$.
\end{itemize}

Note that edges of the former type lead to a vertex $v$ with $\tau(v)=a\pm 1$. Therefore, if $Q$ visits the gadget at $t$ and other vertices of $Z$, then $Q$ follows a rib to an out-port of the gadget at $r$.
As there are two out-ports in the gadget at $r$ and all vertices $v$ of $G_\ell(r)$ with $\tau(v)=a$ lie in the gadget at $r$ or in some leaf of $B_\ell(r)$, we deduce that there are at most 
%three
four nodes of $B_\ell(r)$ whose gadget is intersected by $Q$: if $1\le \alpha < \beta\le q$ are the indices such that $u_\alpha$ and $u_\beta$ are out-ports of the gadget at $r$, then $Q$ may visit a leaf of $B_\ell(r)$ before $u_\alpha$, between $u_\alpha$ and $u_\beta$ and after $u_\beta$. Hence, at most three leaves.
%(those of $r$ and of two leaves of $B_\ell(r)$).

For each $j\in \{1,2\}$, the above argument also applies to $s_j$ and the leaves of $B_\ell(s_j)$:  there are at most four nodes of $B_\ell(s_j)$ where gadgets are intersected by $Q$, which are $s_j$ and at most three leaves. We call these leaves $t,t'$ and $t''$ (and choose them arbitrarily if $Q$ visits less than two gadgets at leaves).

Each of $t, t'$ and $t''$ has two children in $B_\ell$. For each such child $r$, we observed above (as $r\in D$) that $Q$ intersects at most four gadgets of nodes of $B_\ell(r)$.
This shows that among the gadgets of descendants of $s_j$, 
$Q$ intersects at most $1 + 3 + 3 \cdot 2  + 3 \cdot 2 \cdot 3 = 28$ of them. So, in total $Q$ is contained in the union of at most 56 gadgets, as desired.
\end{proof}

\subsection{The induced paths of \texorpdfstring{$G_\ell$}{Gl} are short}
\label{sec:smallip}

The following definition is crucial in the rest of the proof:
a source $s$ is said to be \emph{$Q$-special} if $Q$ contains two vertices $u,v$ such that $u$ is an out-port of the gadget at $s$ and $\pi(v)$ is an internal node of $B_\ell(s)$.

For an induced path $Q$, a $Q$-special source $s$ is an important landmark as, intuitively, it identifies a point where $Q$ enters deeper in the tree-structure formed by the ribs. Because of the ribs attached to the gadget at $s$, the subpath of $Q-\{u\}$ containing $v$ (for $v,u$ as in the definition above) will mostly be restricted to $G_\ell(s)$ and will continue towards $Q$-special sources of smaller rank. This will allow us to bound the length of the path if we can also bound the length between two consecutive $Q$-special sources (see \Cref{lem:Q-special-size}). Notice that the definition of a special source above ignores the part of $Q-\{u\}$ that does not contain $v$, so in the final proof (see \Cref{lem:short}) we will need to consider separately the two subpaths of $Q-\{u\}$.

Actually, due to our different construction, we cannot rely only on special sources as in \cite{defrain2024sparse} so we will introduce the notion of $Q$-reducing sources in order to be able to show that one side of the path is indeed confined to a smaller subgraph (\Cref{lem:source-tau}).

\begin{lemma}\label{lem:Q-special-size}
There is a constant $c_{\ref{lem:Q-special-size}}$ such that the following holds.
Let $Q$ be an induced path in $G_\ell$ such that no source in $B_\ell$ is $Q$-special.
Then $|Q| \leq c_{\ref{lem:Q-special-size}} \cdot \ell$.
\end{lemma}

Here the proof significantly deviates from that of \cite{defrain2024sparse} because of the different gadget that we use.

\begin{proof}
Let $Q = v_1,\dots ,v_q$ be such an induced path. For any two integers $i < j$ such that
\begin{itemize}
    \item $\depth(v_i) < \depth(v_{i+1})$,
    \item $\depth(v_{i+1}) = \dots = \depth(v_{j-1})$, and
    \item $\depth(v_{j-1}) > \depth(v_j)$,
\end{itemize}
we say that $[i, j]$ is a \emph{plateau} of $Q$.
Informally, a plateau captures a local maximum of $\depth$ along the path $Q$. 

\smallskip

The proof of the lemma is split into four claims.

\begin{claim}\label{cl:plateau-structure}
Let $[i, j]$ be a plateau of $Q$.
Then there is a source $s$ such that $v_i$ and $v_j$ are the two out-ports of the gadget at $s$, $v_iv_{i+1}$ and $v_{j-1}v_j$ are ribs and $\tau$ is constant on $\{v_i, \dots, v_j\}$.
\end{claim}

\begin{proof}
Since $[i, j]$ is a plateau, $\depth$ is constant on $\{v_{i+1}, \dots, v_{j-1}\}$.
Since any edge of $G_\ell$ between two gadgets connects gadgets of different depth, $\{v_{i+1}, \dots, v_{j-1}\}$ are contained in a unique gadget, say at node $t$. By definition of $\tau$, this implies that $\tau$ is constant on $\{v_{i+1}, \dots, v_{j-1}\}$.

Since $\depth(v_i) < \depth(v_{i+1})$ and $\depth(v_j) < \depth(v_{j-1})$ the edges $v_iv_{i+1}$ and $v_{j-1}v_j$ are either tree edges or ribs (\Cref{rem:depth-decreasing}).
If they are both tree edges, then $v_i$ and $v_j$ are the two (adjacent) connectors of the gadget at the parent node of $t$; a~contradiction since $Q$ is induced.
Otherwise, if $v_iv_{i+1}$ is a rib and $v_{j-1}v_j$ is a tree edge, then $v_i$ is an out-port of the gadget at a source $s \in V(B_\ell)$ and $v_j$ is a connector of the parent $t'$ of $t$. Observe that $t'$ is an internal node of $B_\ell(s)$. So, $s$ is $Q$-special, a~contradiction.
If $v_iv_{i+1}$ is a tree edge and $v_{j-1}v_j$ is a rib, we reach the same contradiction by symmetry.
Hence we conclude that $v_iv_{i+1}$ and $v_{j-1}v_j$ are both ribs.
In particular, $v_i$ and $v_j$ are out-ports of the gadget at a source $s$.
This implies that $\left(\depth(v_i), \depth(v_{i+1})\right)=\left(\depth(v_j), \depth(v_{j-1})\right)$ is an interval of $\mathcal{N}_\ell$, and thus $\tau$ is constant on $\{v_i, \dots, v_j\}$.\cqed
\end{proof}

\begin{corollary}\label{cor:plateau-disjoints}
Two different plateaux of $Q$ cannot intersect. 
\end{corollary}
\begin{proof}
Let $[i, j]$ and $[k, l]$ be plateaux with $[i, j] \cap [k, l] \neq \emptyset$. By the assumption on the depth, we  must have $j = k$ or $i = l$. But \Cref{cl:plateau-structure} implies that $v_i$, $v_j$, $v_k$ and $v_l$ are all out-ports of the gadget at some source $s$. Since $s$ has only two out-ports, $[i, j] = [k,l]$.
\cqed
\end{proof}

Let $[i, j]$ be a plateau of $Q$. By Claim \ref{cl:plateau-structure}, $v_i$ and $v_j$ are at the same depth, say at depth $d$.
If $\depth(v_k)>d$ for any $k \in \intv{1}{i-1}$, then $[i, j]$ is called a \emph{decreasing plateau}. If $\depth(v_k)>d$ for any $k\in \intv{j+1}{q}$, then $[i,j]$ is called an  \emph{increasing plateau}.

\begin{claim}\label{cl:inc-plateau}
Any plateau $[i, j]$  of $Q$ is a decreasing or increasing plateau.
\end{claim}
\begin{proof}
Since $[i, j]$ is a plateau of $Q$, an application of \Cref{cl:plateau-structure} implies that $v_i$ and $v_j$ are out-ports of a same gadget $J$ at a source $s$.
In particular, $v_{i-1}$ and $v_{j+1}$ are not vertices of $J$ since otherwise, $v_{i-1}v_j$ or $v_iv_{j+1}$ would be an edge.

Hence, $v_{i-1}v_i$ and $v_jv_{j+1}$ are either tree edges or ribs.
If they are both tree edges, then $v_i$ and $v_j$ are two adjacent connectors of the gadget at the parent of $s$; a~contradiction since $Q$ is induced.
Hence  $v_{i-1}v_i$ is a rib, or $v_jv_{j+1}$ is a rib.
We show that $\depth(v_i) < \depth(v_k)$ for any $k \in \intv{1}{i-1}$ when $v_{i-1}v_i$ is a rib. The proof that $\depth(v_i) < \depth(v_k)$ for any $k \in \intv{j+1}{q}$ when $v_jv_{j+1}$ is a rib is symmetric.

Assume for the sake of contradiction that $v_{i-1}v_i$ is a rib and that there exists an index $k \in \intv{1}{i-1}$ such that $\depth(v_i) \ge \depth (v_k)$, and take $k$ maximum in $\{1, \dots, i-1\}$ with this property.
This implies that $\pi(v_i)$ is an ancestor of $\pi(v_{k+1})$ (possibly $\pi(v_i) = \pi(v_{k+1})$): indeed every edge of $G_\ell$ is between $\preceq$-comparable vertices, and since $v_{k+1}, \dots ,v_i$ is a path, some vertex among $\{\pi(v_{k+1}), \dots, \pi(v_i)\}$ is a common ancestor in $B_\ell$ of all the others (\Cref{lem:depth-root}); and since $\depth(v_i)$ is minimal, $\pi(v_i)$ is this common ancestor, hence it is an ancestor of $\pi(v_{k+1})$.
 
In the case where $\depth(v_i) > \depth(v_k)$, the vertex $v_k$ is an out-port of a gadget at a source $s'$. The node $\pi(v_i)$ is in $B_\ell(s')$ since $\depth(v_{k}) < \depth(v_i) < \depth(v_{k+1})$ and $\pi(v_{i})$ is a ancestor of $\pi(v_{k+1})$. Hence, the source $s'$ is $Q$-special; a~contradiction.
Otherwise, $\depth(v_k) = \depth(v_i)$; we know that $\pi(v_i)$ is an ancestor of $\pi(v_{k+1})$, so $v_k, v_i$ and $v_j$ are all in the gadget $J$. Since $v_i$ and $v_j$ are the two out-ports of $J$ and $\depth(v_{k+1}) > \depth(v_k)$, we have that $v_k$ is a connector of $J$, and $v_kv_{k+1}$ is a tree edge. Hence, $s$ is $Q$-special, a~contradiction.\cqed
\end{proof}

\begin{claim}\label{cl:bimon-plateau}
If $[i, j]$ is a decreasing plateau, and $[k, l]$ an increasing plateau of $Q$, then $j < k$.
\end{claim}
\begin{proof}
Indeed, let $a < b < c < d$ such that $[a, b]$ is a increasing plateau and $[c, d]$ an decreasing one.
Then since $a < c$ we have by the decreasing property (\Cref{cl:inc-plateau}) of $[c, d]$ applied with $k := a$ and $i := c$ that $\depth(v_a) > \depth(v_c)$; and by the increasing property of $[a, b]$ applied with $k := c$ and $i := a$ that $\depth(v_c) > \depth(v_a)$; a~contradiction.\cqed
\end{proof}

\begin{claim}\label{cl:tau-bimon}
There exists a constant $c_{\ref{cl:tau-bimon}}$ (depending only on the gadget size in the construction) and an integer $\alpha \in \intv{1}{q}$ such that for any pair $1\le i< j\le q$,
\begin{itemize}
    \item if $ i < j \le \alpha - c_{\ref{cl:tau-bimon}}$, then $\tau(v_i) \ge \tau(v_j)$, and 
    \item if $\alpha + c_{\ref{cl:tau-bimon}} \le i < j$, then $\tau(v_i) \le \tau(v_j)$.
\end{itemize}
\end{claim}
\begin{proof}
By \Cref{cor:plateau-disjoints}, there is a finite sequence of intervals $([a_i, b_i])_{1 \le i \le p}$ containing exactly the plateaux of $Q$ and such that $b_i \le a_{i+1}$ for any $i$. \Cref{cl:bimon-plateau} asserts that initial plateaux in the sequence are decreasing while final plateaux are increasing, hence there is an index $m\in \intv{1}{p}$ such that $\depth(v_{a_1}) > \dots > \depth(v_{a_m})$, and $\depth(v_{a_{m+1}}) < \dots < \depth(v_{a_p})$.

Since any local maximum of $\depth$ is in a plateau, the depth of vertices along $Q$ is  decreasing between two decreasing plateaux and increasing between two increasing plateaux.
More precisely, there is an integer $k \in \intv{b_{m}}{a_{m+1}}$ such that for any $i \in \intv{1}{k-1}$, either $\depth(v_i) \ge  \depth(v_{i+1})$ or $i$ and $i+1$ are in some decreasing plateau; and symmetrically for $i \in \intv{k+1}{q}$ either $\depth(v_i) \le \depth(v_{i+1})$ or $i$ and $i+1$ are in some increasing plateau. 

We now prove that $\tau$ is decreasing on $v_1 ,\dots, v_{k-c_{\ref{cl:tau-bimon}}}$ (where $c_{\ref{cl:tau-bimon}}$ is the size of a gadget) and that $\tau$ is increasing on $v_{k+c_{\ref{cl:tau-bimon}}}, \dots ,v_{q}$. Note that proving the increasing part proves the decreasing part by reversing the order of $Q$.

The proof goes by contradiction. Let us assume that there is a (smallest) integer $i \in \intv{k+c_{\ref{cl:tau-bimon}}}{q-1}$ such that $\tau(v_{i}) > \tau(v_{i+1})$, and let $s = \pi(i)$.
Since $\tau$ is constant on plateaux, $i$ and $i+1$ are not both in some plateau, hence $\depth(v_i) \le \depth(v_{i+1})$.

By Lemma~\ref{lem:tau-decreasing} and since $\depth(v_i)\leq \depth(v_{i+1})$, the node $s$ is a source, $v_iv_{i+1}$ is a tree edge so $v_i$ is a connector of the gadget at $s$ and $v_{i+1}$ is an out-port of the gadget at a child of $s$. Note that, since $s$ is not $Q$-special, no out-ports of $s$ are in $Q$.

We now look at the part of the path between $v_k$ and $v_i$.

Assume there exists a largest integer $i' \in \intv{k+1}{i}$  such that $\pi(v_{i'}) \neq \pi(v_i)$. The edge $v_{i'}v_{i'+1}$ is either a tree edge or a rib, but $Q$ does not contain an out-port of $s$, hence $v_{i'+1}$ is a connector of $s$. This implies $\depth(v_{i'}) > \depth(v_{i'+1})$. Hence, there exists a plateau $[a, b]$ containing $i'$ and $i'+1$. This is impossible by \Cref{cl:plateau-structure} since $v_{i'+1}$ is not an out-port.

Hence such an $i'$ does not exist, thus $\pi(v_k) = \dots = \pi(v_i)$, and the gadget at $\pi(v_k)$ contains $c_{\ref{cl:tau-bimon}}+1$ vertices; a~contradiction.\cqed
\end{proof}

We are now ready to conclude the proof of Lemma \ref{lem:Q-special-size}. Recall that the function $\tau$ has values in $\intv{1}{\ell+1}$. By virtue of \Cref{cl:tau-bimon}, there is an integer $k$ such that $\tau$ is non-increasing on
$v_1, \dots ,v_{k-c_{\ref{cl:tau-bimon}}}$ and non-decreasing on $v_{k+c_{\ref{cl:tau-bimon}}}, \dots ,v_q$.
Furthermore, $\tau$ does not keep the same value on more than $c_{\ref{lem:tau-constant}}$ consecutive vertices (Lemma~\ref{lem:tau-constant}). 
We conclude that $Q$ has order at most $2c_{\ref{lem:tau-constant}} (\ell+1) + 2c_{\ref{cl:tau-bimon}}$, which is bounded from above by $c\cdot \ell$ for some constant~$c$, as $\ell\geq 1$. This concludes the proof of Lemma \ref{lem:Q-special-size}.
\end{proof}

We now show that if an induced path $Q$ visits some out-port of the gadget at a source~$s$, as well as a gadget at an internal node of $B_\ell(s)$, then a certain suffix of $Q$ remains inside $V(G_\ell(s))$. 

\begin{lemma}\label{lem:source-tau}
Let $Q=v_1,\dots, v_q$ be an induced path of $G_\ell$ such that $v_1$ is an out-port of the gadget at a source $s$. If $s$ is a $Q$-special source, then:
\begin{itemize}
    \item\label{it:reducing-source} If $Q$ does not contain the second out-port of $s$, then $V(Q) \subseteq V(G_\ell(s))$.
    \item\label{it:bi-special} Otherwise, let $j > 1$ such that $v_j$ is the second out-port of $s$, and consider the two subpaths  $Q_1=v_1,\ldots,v_j$ and $Q_2=v_j,\ldots,v_q$ of $Q$. Then, $s$ is not $Q_1$-special, and $s$ is $Q_2$-special (and in particular $V(Q_2) \subseteq V(G_\ell(s))$, by the first item).
\end{itemize}
\end{lemma}
\begin{proof}
Recall that the definition of $\Int(s)$ and $\Out(s)$ is given shortly before the statement of \Cref{lem:separation}.
%Let $I(s)$ be the set of vertices $v$ such that $\pi(v)$ is in the interior of $B_\ell(s)$,
%and let $O(s)$ be the set of vertices $v$ such that $\pi(v) \not\in B_\ell(s)$.
Recall also that $s$ has two out-ports: the left one that we denote by $u_L$ and the right one that we refer to as $u_R$.
Let $i\in \intv{1}{q}$ be such that $v_i \in \Int(s)$, which exists since $s$ is $Q$-special.

Assume first that $Q$ does not contain the second out-port of the gadget at the source~$s$. By \Cref{lem:separation}, any subpath of $Q$ from $v_i$ to $\Out(s)$ intersects a neighbor of $v_1$. Since $Q$ is induced, $Q$ cannot have a vertex in $\Out(s)$, hence $V(Q) \subset V(B_\ell(s))$.

Suppose now that there is an index $j\in \intv{1}{q}$ such that $v_j$ is the second out-port at the source $s$. If $1 < i < j$, then either $v_1$ or $v_j$ is the left out-port $u_L$.
Assume $v_1 = u_L$ (resp. assume $v_j = u_L$). Then, by \Cref{lem:separation}, $N(v_1)$ (resp.\ $N(v_j)$) separates $v_i$ from $\Out(s) \cup\{ u_R\}$, so $v_j$ cannot be $u_R$ (resp.\ $v_1$ cannot be $u_R$) otherwise $Q$ is not induced; a contradiction towards the definitions of $v_1$ and $v_j$. Hence, we have $1 < j < i$ and the second item is proved.
\end{proof}

In the setting of \Cref{lem:source-tau}, if $Q$ does not contain the second out-port of $s$, then we say that  $s$ is a $Q$-\emph{reducing source}. In this case, \Cref{lem:source-tau} states that $V(Q) \subseteq V(G_\ell(s))$.

\begin{lemma}\label{lem:reducing-source-structure}
Let $Q = v_1, \dots, v_q$ be an induced path of $G_\ell$, and for any $1\le i\le q$, define the subpath $Q_i=v_i,\ldots,v_q$ of $Q$. For any $a \in \intv{1}{\ell}$, there is at most one index $1\le i\le q$ such that $\pi(v_i)$ is a $Q_i$-reducing source of rank $a$.
\end{lemma}
\begin{proof}
Indeed, let $a \in \{1, \dots, \ell\}$. Assume for the sake of contradiction that there exist $i < j$ such that $\pi(v_i)$ and $\pi(v_j)$ are respectively $Q_i$-reducing and $Q_j$-reducing sources of rank~$a$. By definition of a reducing source and by \Cref{lem:source-tau}, we have $V(Q_i) \subseteq V(G_\ell(\pi(v_i)))$. In particular, $v_j \in V(G_\ell(\pi(v_i)))$. But since the rank of $\pi(v_j)$ is the same as the rank of $\pi(v_i)$, and since $\pi(v_j)$ is a source, we have $\pi(v_i) = \pi(v_j)$, and thus $v_j\in Q_i$ is the second out-port of $\pi(v_i)$. This implies that $\pi(v_i)$ is not $Q_i$-reducing; a contradiction.
\end{proof}

We are now ready to prove that all induced paths in $G_\ell$ are short.

\begin{lemma}\label{lem:short}
There is a constant $c$ such that for any induced path $Q = v_1, \dots, v_q$  of $G_\ell$, $q \le c \ell^2$.
\end{lemma}
\begin{proof}
Let $Q = v_1, \dots, v_q$ be an induced path of $G_\ell$.
Note that for any $i \in \intv{1}{q}$, $\pi(v_i)$ is a $Q$-special source if and only if it is either a $(v_i, \dots ,v_q)$-special source or a $(v_i, v_{i-1}, \dots ,v_1)$-special source.
By \Cref{lem:reducing-source-structure}, for any $a \in \intv{1}{\ell}$ there is at most one index $i_a$ such that $\pi(v_{i_a})$ is a $(v_{i_a}, \dots ,v_q)$-reducing source of rank $a$. Let $i_1 < \dots < i_p$ with $p \le \ell$ denote the sequence of such indices. Note that $\rank(\pi(v_{i_j})) > \rank(\pi(v_{i_{j+1}}))$: indeed by \Cref{lem:source-tau} we have $\pi(v_{i_{j+1}}) \in B_\ell(\pi(v_{i_j}))$.

It follows that for any $k\in\{i_j + 1, \dots, i_{j+1}\}$,
$\pi(v_k)$ cannot be a $(v_{i_j + 1} ,\dots, v_{i_{j+1}})$-special source:
indeed for any such $k$ we have  $v_k \in G_\ell(\pi(v_{i_j}))$ by \Cref{lem:source-tau}, and thus 
$\rank(\pi(v_{i_j})) \ge \rank(\pi(v_k))$,
hence $\pi(v_k)$ cannot be a $(v_k ,v_{k-1}, \dots, v_{i_j})$-special source,
and if $\pi(v_k)$ was a $(v_k ,\dots, v_{i_{j+1}})$-special source,
we would have \[\rank(v_k) > \rank(v_{i_{j+1}}),\]
a~contradiction to the definition of $v_{i_{j+1}}$. 
Hence by \Cref{lem:Q-special-size}, $i_{j+1} - i_j \le c_{\ref{lem:Q-special-size}} \ell$, and thus $q - i_1 \le  c_{\ref{lem:Q-special-size}} \ell^2$, and for any vertex $v_a$ with $a \le i_1$, $\pi(v_a)$ is not a $(v_a, \dots ,v_{i_1})$-special source. 

A completely symmetric argument on the path $v_{i_1}, v_{i_1 - 1}, \dots, v_1$ allows us to bound the index $j$
such that $\pi(v_{j})$ is a $(v_j, v_{j - 1}, \dots, v_1)$-special source of maximal rank, by $c_{\ref{lem:Q-special-size}} \ell^2$, and ensure that for any $v_a$ with $a > j$, $\pi(v_a)$ is not a $(v_a, v_{a - 1} ,\dots ,v_j)$-special source.
In particular, no source is $(v_j, v_{j+1}, \dots, v_{i_1})$-special, and by \Cref{lem:Q-special-size}, $|v_j, v_{j+1} ,\dots, v_{i_1}| \le c_{\ref{lem:Q-special-size}}\cdot \ell$.
We can now deduce the desired bound:
\[
|Q|
\le |v_1, \dots, v_j| + |v_j, \dots, v_{i_1}| + |v_{i_1}, \dots, v_q|
\le c_{\ref{lem:Q-special-size}} \ell^2 + c_{\ref{lem:Q-special-size}} \ell + c_{\ref{lem:Q-special-size}} \ell^2
\le 3 c_{\ref{lem:Q-special-size}} \ell^2.
\]
\end{proof}

Recall \Cref{lem:loccons}, which states that constellations are precisely the ordered forests of left and right stars whose stars can be ordered as $S_1,\ldots,S_t$, in such a way that that for any $i<j$, the center $c_i$ of $S_i$ is outside $S_j$ (meaning that $c_i$ precedes or succeeds the vertex set of $S_j$ in the ordered forest).

\medskip

We are now ready to prove \Cref{th:loglogbis}, restated below for convenience.
\loglogbis*
\begin{proof}
It suffices to show that for any ordered graph $H$ which is not a constellation and for every $\ell$, the graph $G_\ell$, ordered along some Hamiltonian path $P$, avoids $H$ as a pattern. In particular, it is enough to show that $G_\ell-E(P)$ (ordered along $P$) is a constellation.

Note that $G_\ell$ has a natural Hamiltonian path $P$ starting at the left out-port of the root gadget and ending at the right out-port of the root gadget, which does not use any rib (see Figure \ref{fig:hamgell} for an illustration). 

\begin{figure}[htb]
  \centering
    \includegraphics[scale=1.1]{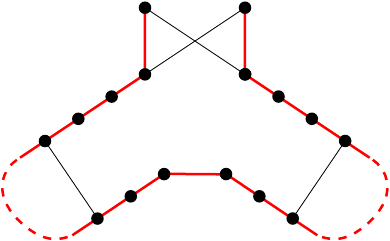}
  \caption{The Hamiltonian path $P$ in $G_\ell$ (in red). The dashed parts consist of two Hamiltonian paths defined inductively in the left and right subtrees of the root gadget.\label{fig:hamgell}}
\end{figure}

We want to show that $G_\ell-E(P)$ (ordered according to $P$) is a constellation, i.e., it is a star forest in which each star is either a right or left star, and property $(\star)$ of Lemma \ref{lem:loccons} is satisfied.

Note that $G_\ell-E(P)$ indeed consists of a union  of disjoint stars: the edges between connectors form a matching and all the other stars consist of the ribs originating from some out-port $v$, together with a single edge of the gadget containing the out-port $v$. Hence, since the ribs go from some gadget $\pi(s)\ni v$ to the vertices of $G_\ell(s)$, and since the out-ports of the gadget at $s$ are the first and last vertex of $G_\ell(s)$ visited by the Hamiltonian path $P$, $v$ is either the smallest or largest vertex of the star. So, each component of $G_\ell-E(P)$ is indeed a left or right star.

Now, consider two stars $S_1$ with center $c_1$ and $S_2$ with center $c_2$, with $c_1 \neq c_2$, and  assume that $\depth(c_1) \le  \depth(c_2)$. Then,  $c_1$ is outside of  $G_\ell(\pi(c_2))$, and thus also outside of $S_2$. This shows that if we order the stars of $G_\ell-E(P)$ by increasing depth of their centers, property $(\star)$ of Lemma \ref{lem:loccons} is satisfied, and thus $G_\ell-E(P)$ is a constellation.
\end{proof}

\section{Discussion}\label{sec:open}

By \Cref{cor:topomin}, every graph without $K_t$ as a topological minor, and which contains an $n$-vertex path, also contains an induced path of order $(\log n)^{\Omega\left (1/t(\log t)^2\right )}$. A natural question is whether the exponent can be improved. The same can also be asked more specifically for graphs without $K_t$-minor. An upper bound of  $O((\log n)^{2/(t-1)})$ was proved in \cite{esperet2017long}, holding for the more restrictive class of graphs of treewidth less than $t$.

\medskip

In \cite{I}, we have proved the following  dichotomies, revealing jumps in the growth rate of~$g_H$.
\begin{corollary}[\cite{I}]\label{cor:dicho}
Let $H$ be an ordered graph.
\begin{enumerate}
    \item $g_H(n) = n^{\Omega(1)}$ if and only if $H$ is a subgraph of a non-crossing matching, and otherwise $g_H(n) = O(\log n)$;
    \item\label{e:obip} $g_H(n) = (\log \log n)^{\Omega(1)}$ if $H$ is bipartite, with one partite set preceding the other in the order;
    \item \label{it:logloglog} $g_H(n) = (\log \log \log n)^{\Omega(1)}$ if and only if $H$ is a subgraph of the ordered half-graph, and otherwise $g_H(n) = O(1)$;
\end{enumerate}
\end{corollary}

The following direct consequence of Theorems  \ref{th:peel_omega} and \ref{th:loglogsqr} can now be added to the list.

\begin{corollary}\label{cor:dicho2}
Let $H$ be an ordered graph. Then, $g_H(n) = (\log n)^{\Omega(1)}$ if and only if $H$ is a constellation, and otherwise $g_H(n) = O((\log \log n)^2)$.
\end{corollary}

A natural question is whether the triple logarithm is necessary in \Cref{it:logloglog} of \Cref{cor:dicho}. It might very well be the case that this bound can be replaced by $(\log \log n)^{\Omega(1)}$. This would give a complete picture of the possible complexities of the function $g_H$.

\bibliographystyle{alpha}
\bibliography{references}

\appendix

\section{Stars: proof of the crucial inequalities}
\label{sec:stars-proofs}

In this section, we prove \Cref{lem:mono} and \Cref{lem:bounds}, which are  crucial steps in the proof of \Cref{th:peel} about induced paths in graphs avoiding a constellation.
Recall that the functions $\varphi$, $\gamma$, $\eta$, $f$, $g$, and $h$ are given in Definitions~\ref{def:phigamma} and~\ref{def:functions}.

\medskip

We start with \Cref{lem:mono}, which we restate below.

\lemmono*

\begin{proof}
\noindent \textbf{Monotonicity of $f$.}
Since $p/2$ appears on both sides of the inequality, we can equivalently prove the following.
\begin{align}
\label{eq:f-dec}\left( \log_{r+1} n \right)^{\varphi(t)} - 4 \pow \frac{1}{\varphi(t-1) - \eta(t)} \le
\left( \log_{r+1} n \right)^{\varphi(t-1)} - 4 \pow \frac{1}{\varphi(t-2) - \eta(t-1)}.
\end{align}
Since $\varphi(t-1) > \varphi(t)$, the function
$n \mapsto \left( \log_{r+1} n \right)^{\varphi(t)} - \left( \log_{r+1} n \right)^{\varphi(t-1)}$ is decreasing.
Hence it is sufficient to prove the inequality for the lowest possible value of $n$, which is $\log_{r+1} n = 4 \pow \frac{1}{\varphi(t) \cdot(\varphi(t-1) - \eta(t))}$, and \Cref{eq:f-dec} becomes

\[
4 \pow \frac{1}{\varphi(t-1) - \eta(t)} - 4 \pow \frac{1}{\varphi(t-1) - \eta(t)} \le
4 \pow \frac{\varphi(t-1)}{\varphi(t) (\varphi(t-1) - \eta(t))} - 4 \pow \frac{1}{\varphi(t-2) - \eta(t-1)}.
\]

Hence, \Cref{eq:f-dec} boils down to
\[
\frac{1}{\varphi(t-2) - \eta(t-1)} \le
\frac{\varphi(t-1)}{\varphi(t) (\varphi(t-1) - \eta(t))}.
\]
Since $\varphi(t-1) \ge \varphi(t)$, and since $\varphi(t-2) - \eta(t-1) > \varphi(t-1) - \eta(t)$ (by \Cref{def:phigamma}), the following inequalities holds.
\[
\frac{\varphi(t-1)}{\varphi(t) (\varphi(t-1) - \eta(t))}
\ge \frac{1}{\varphi(t-1) - \eta(t)}
\ge \frac{1}{\varphi(t-2) - \eta(t-1)}
\]
And thus \Cref{eq:f-dec} holds.

\medskip

\noindent \textbf{Monotonicity of  $h$.}
We follow the exact same reasoning as for the function $f$ above.
In particular, we want to prove
\begin{align}
\label{eq:h-dec}\left( \log_{r+1} n \right)^{\eta(t)} - 4 \pow \frac{1}{\varphi(t-2) - \eta(t-1)} \le
\left( \log_{r+1} n \right)^{\eta(t-1)} - 4 \pow \frac{1}{\varphi(t-3) - \eta(t-2)}.
\end{align}

Recall that according to \Cref{def:phigamma}, $\varphi(t-2) - \eta(t-1) > \varphi(t-1) - \eta(t)$ and that $\eta(t) > \varphi(t)$. Hence, our assumption on $n$ implies $\log_{r+1} n \ge 4 \pow \frac{1}{\eta(t) \cdot(\varphi(t-2) - \eta(t))}$. Thus, it is sufficient to prove the following inequality to prove \Cref{eq:h-dec}.

\begin{align*}
4 \pow \frac{1}{\varphi(t-2) - \eta(t)} - 4 \pow \frac{1}{\varphi(t-2) - \eta(t-1)}\hspace{6cm}\\
\hspace{5cm}\le
4 \pow \frac{\eta(t-1)}{\eta(t) (\varphi(t-2) - \eta(t-1))} - 4 \pow \frac{1}{\varphi(t-3) - \eta(t-2)}.
\end{align*}
Since $\eta(t-1) \ge \eta(t)$ (by \Cref{def:phigamma}), the left-hand side is negative and thus in order to prove \Cref{eq:h-dec} it is sufficient to show
\[
\frac{1}{\varphi(t-3) - \eta(t-2)} \le
\frac{\eta(t-1)}{\eta(t) (\varphi(t-2) - \eta(t-1))}.
\]
Since $\eta(t-1) \ge \eta(t)$, and since $\varphi(t-3) - \eta(t-2) > \varphi(t-2) - \eta(t-1)$, the following inequalities holds.
\[
\frac{\eta(t-1)}{\eta(t) (\varphi(t-2) - \eta(t-1))}
\ge \frac{1}{\varphi(t-2) - \eta(t-1)}
\ge \frac{1}{\varphi(t-3) - \eta(t-2)},
\]
and thus \Cref{eq:h-dec} holds.

\medskip

\noindent \textbf{Monotonicity of  $g$.}
We will prove the equivalent inequality
\begin{align}
\label{eq:g-dec} \frac{g(n, t-1, p)}{g(n, t, p)} \ge 1.
\end{align}
Let $\ell = \log_{r+1} n$. \Cref{eq:g-dec} is true if and only if $\log_{6(r+1)}\left(\frac{g(n, t-1, p)}{g(n, t, p)}\right) \ge 0$, which after simplification gives the following.
\begin{align}
\label{eq:g-dec2} \ell^{\gamma(t)} \cdot (3 \ell^{\varphi(t)} - p) \ge \ell^{\gamma(t-1)} \cdot (3 \ell^{\varphi(t-1)} - p).
\end{align}
Recall that we assume 
$p \leq 2 \cdot \ell^{\varphi(t)}$.
Since $\gamma(t) > \gamma(t-1)$, the map $p \mapsto p \cdot \left(\ell^{\gamma(t)} - \ell^{\gamma(t-1)} \right)$ is increasing. Hence, it is sufficient to prove the inequality when $p=2\ell^{\varphi(t)}$, i.e., to show
\[
\ell^{\gamma(t) + \varphi(t)} \ge \ell^{\gamma(t-1)} (3\ell^{\varphi(t-1)} - 2\ell^{\varphi(t)}).
\]
From \Cref{def:phigamma} we have $\varphi(t) > \varphi(t-1)$ so the above inequality is implied by the following, where we replaced $2\ell^{\varphi(t)}$ with $2\ell^{\varphi(t-1)}$:
\[
\ell^{\gamma(t) + \varphi(t)} \ge \ell^{\gamma(t-1) + \varphi(t-1)}.
\]
This holds since $\gamma(t) + \varphi(t) \ge \gamma(t-1) + \varphi(t-1)$, as can be deduced from \Cref{def:phigamma}.
\end{proof}

Before proving \Cref{lem:bounds}, we will need two preliminary lemmas.
%%% Computations %%%

\begin{lemma}\label{lem:log inequality}
For any $c_0 \in (0,1)$, $c_1>0$, $\ell \ge \max\{1,c_1^{1/(1-c_0)}\}$, and $x \leq 1 - c_0 - \log_\ell (2c_1)$ we have $(\ell - c_1 \cdot \ell^{c_0})^x \geq \ell^x - 1/2$.
\end{lemma}
\begin{proof}

The statement is equivalent (as $\ell$ is positive) to the following inequality: \begin{align}\label{eq:cuisine}
(1 - c_1 \cdot \ell^{c_0 - 1})^x \geq 1 - \ell^{-x}/2.
\end{align}

By assumption, $c_1^{1/(1-c_0)} \le \ell$ and thus (since $1 - c_0$ is positive) $c_1\le \ell^{1-c_0}$. It follows that $c_1\ell^{c_0-1}\le 1$ and so $1 -  c_1 \ell^{c_0  - 1} \in [0, 1]$. Hence,

\begin{align*}
    \left (1- c_1 \ell^{c_0-1}\right )^x
    &\geq 1-c_1\ell^{c_0-1} &\text{as}\ x \leq 1 \ \text{and} \ 1- c_1 \ell^{c_0-1} \in [0, 1]\\
    &\geq 1-\frac{\ell^{-(1-c_0-\log_\ell (2c_1))}}{2}&\text{as}\ \ell\ge c_1^{1/(1-c_0)}\\
    &\geq 1-\frac{\ell^{-x}}{2} &\text{as}\ x \leq 1 - c_0 - \log_\ell (2c_1) \ \text{and}\ \ell \ge  1,
\end{align*}
which is \cref{eq:cuisine}.
\end{proof}

\begin{lemma}\label{lem:lowbdstretch}
Let $r\ge 1$, $t\ge 1$, $p\ge 0$, $n\ge 1$ be integers and let $\ell = \log_{r+1} n$. If $\ell\geq 2^{1/\varphi(t)}$ then we have:
\[
s(n, t, p) \geqslant
\frac{n} {(6(r+1)) \pow \left(2 \ell^{\gamma(t-1)} \left( 3\ell^{\varphi(t-1)} - p\right) + 1 \right)}.
\]
\end{lemma}
\begin{proof}
By \hyperref[def:functions]{definition}, we have:
\begin{align*}
s(n,t,p) &= \frac{g(n/3, t-1, p) - 1}{2r + 1} \\
    &= \frac{1}{2r+1} \cdot \left( \frac{n/3}{(6(r+1)) \pow \left(2(\log_{r+1} (n/3))^{\gamma(t-1)} \cdot (3(\log_{r+1} (n/3))^{\varphi(t-1)} - p) \right)} - 1 \right)\\
    &\geqslant \frac{1}{3 (2r+1)} \cdot \left( \frac{n}{(6(r+1)) \pow \left(2\ell^{\gamma(t-1)} \cdot (3\ell^{\varphi(t-1)} - p) \right)} - 3 \right)\\
    &=  \frac{\frac{n}{3(2r+1)} - \frac{1}{2r+1}(6(r+1)) \pow \left(2\ell^{\gamma(t-1)} \cdot (3\ell^{\varphi(t-1)} - p) \right)}{(6(r+1)) \pow \left(2\ell^{\gamma(t-1)} \cdot (3\ell^{\varphi(t-1)} - p) \right)}. 
\end{align*}
Hence it would be sufficient to prove:
\begin{align}
\label{eq:sub}
    \frac{n}{3(2r+1)} - \frac{1}{2r+1}(6(r+1)) \pow \left(2\ell^{\gamma(t-1)} \cdot (3\ell^{\varphi(t-1)} - p) \right)
    \geqslant \frac{n}{6(r+1)}.
\end{align}
Note that
$\frac{n}{3(2r+1)} - \frac{n}{6(r+1)}
%= \frac{6(r+1)n - 3(2r+1)n}{3(2r +1)6(r+1)}
= \frac{n}{(2r +1) \cdot 6(r+1)}
= \frac{1}{2r+1}(6(r+1)) \pow \left(\log_{6(r+1)} n - 1 \right)$.
So \Cref{eq:sub} can be rewritten as:
\begin{align*}
\log_{6(r+1)} n - 1 \geqslant 2\ell^{\gamma(t-1)} \cdot (3\ell^{\varphi(t-1)} - p).
\end{align*}
We now prove the above inequality.
\begin{align*}
\log_{6(r+1)} n - 1 
&= \frac{\log(r+1)}{\log(r+1) + \log 6} \log_{r+1} n- 1\\
&\geqslant \frac{1}{1 + 3} \ell - 1 &\hspace{-3cm}\text{as} \ \ell = \log_{r+1} n, \  r + 1 \geqslant 2, \ \text{and} \ \log 6 \le 3\\
&\ge \frac18 \ell &\text{as} \ \ell \ge 8,\ \text{by the assumption on}\ n\\
&\ge \frac18 \ell^{7 \cdot \varphi(t-1)} \ell^{\gamma(t-1) + \varphi(t-1)} &\\
&&\hspace{-8cm}\text{as by \Cref*{def:phigamma}.\eqref{eq:unmoins}}, \ \ell^{1 - \gamma(t-1) - \varphi(t-1)} \ge \ell^{7 \cdot \varphi(t-1)}\\
&\ge \frac{128}{8} \ell^{\gamma(t-1) + \varphi(t-1)} &\text{by assumption on $\ell$}\\
&\geqslant 2 \ell^{\gamma(t-1)} \cdot (3\ell^{\varphi(t-1)} - p). &\text{as} \ p \ge 0.
\end{align*}
This concludes the proof.
\end{proof}

We are now ready to prove \Cref{lem:bounds}, restated hereafter for convenience.

\lembounds*
\begin{proof}

For the sake of readability in the upcoming equations, we define $\ell = \log_{r+1} n$.
Note that assumption \eqref{cond:bign} implies $\ell^{\varphi(t)} \geqslant 2$.
% Note we always have $\ell \geqslant 2^{1/\varphi(t)} \geqslant 2^{(2 \cdot 1 + 1)^2} \geqslant 512$.

\medskip

\noindent {\bf Proof of \eqref{eq:recursionf}.}
We want to prove \[\big(\log_{r+1} s(n, t, p)\big)^{\varphi(t)} - \tfrac{p+1}{2} - 4^{\frac{1}{ \varphi(t-1) - \eta(t)}} \geqslant (\log_{r+1} n)^{\varphi(t)} - \tfrac{p}{2} - 4^{\frac{1}{ \varphi(t-1) - \eta(t)}} - 1,\] or equivalently:
\[
(\log_{r+1} s(n, t, p))^{\varphi(t)} \geqslant \ell^{\varphi(t)} - 1/2.
\]
We have the following:
\begin{align*}
&\big(\log_{r+1} s(n, t, p)\big)^{\varphi(t)}\\
& \geqslant \left( \ell - \log_{r+1} \left( (6(r+1)) \pow \left(2 \ell^{\gamma(t-1)}
\left( 3\ell^{\varphi(t-1)} - p\right) + 1 \right) \right) \right)^{\varphi(t)} &\text{using \Cref{lem:lowbdstretch}}\\
& \geqslant \left( \ell - 6 \log_{r+1}(6(r+1))\cdot \ell^{\gamma(t-1)} \cdot \ell^{\varphi(t-1)} - \log_{r+1}(6(r+1)) \right)^{\varphi(t)} &\text{because}\ p\ge 0\\
& \geqslant
\left( \ell - 7 \log_{r+1}(6(r+1)) \cdot \ell^{\gamma(t-1) + \varphi(t-1)}\right)^{\varphi(t)} &\hspace{-1cm}\text{as } \ell^{\gamma(t-1) + \varphi(t-1)} \geqslant \ell^{\varphi(t)} \geqslant 2.
\end{align*}
We will apply \Cref{lem:log inequality} with $c_0 := \gamma(t-1) + \varphi(t-1)$, $c_1 := 7 \log_{r+1}(6(r+1))$ and $x := \varphi(t)$ to conclude the proof.
Hence we need to satisfy the requirements of the lemma:
\begin{itemize}
    \item $c_0 \in (0, 1)$ since $1 > \gamma(t) \geqslant \gamma(t-1) + \varphi(t-1) > 0$ by \Cref{def:phigamma};
    \item $\max\big(1, c_1^{1/(1 - c_0)}\big) \leqslant \ell$. For this note that $1 - c_0 = 1 - \gamma(t-1) - \varphi(t-1)\ge 7 \varphi(t-1)$ (where the last inequality follows from \Cref{def:phigamma}).
    Note also that $c_1 = 7 \log_{r+1}(6(r+1)) \leqslant 7+7\log 6\leqslant 2^5$. Hence, $c_1^{1/(1 - c_0)}\le c_1^{1/7} \leqslant 2^{\frac{5}{7}}\leqslant \ell$, as desired.
    \item $\varphi(t) \leqslant 1 - c_0 - \log_\ell(2c_1)$.  Observe that
      \[\log_\ell(2c_1) = \frac{\log(2c_1)}{\log \ell} \leqslant
      \frac{\log(2c_1)}{1/\varphi(t)} = \log(2c_1)\varphi(t)\le 6\varphi(t)\leqslant 6\varphi(t-1).\]
    Furthermore, as observed above $1 - c_0 \geqslant 7 \varphi(t-1)$, thus $    1 - c_0 - \log(2c_1) - \varphi(t) \geqslant 0$, as desired.
\end{itemize}

Hence by an application of \Cref{lem:log inequality} we have
\[
\left( \ell - 7 \log_{r+1}(6(r+1)) \cdot \ell^{\gamma(t-1) +
  \varphi(t-1)}\right)^{\varphi(t)}
=  \left( \ell - c_1 \cdot
\ell^{c_0}\right)^{\varphi(t)}  \geqslant \ell^{\varphi(t)} - 1/2.
\]

Thus, $f\big( s(n, t, p), t, p+1\big) \geqslant f(n, t, p) - 1$, as desired.

\medskip

\noindent {\bf Proof of \eqref{eq:recursionh}.}
This calculation is very similar to the previous one, replacing $\varphi(t)$ by $\eta(t)$.
We want to prove 
\[\big(\log_{r+1} s(n, t, p)\big)^{\eta(t)} + \frac{p+1}{2} - 4^{\frac{1}{ \varphi(t) - \eta(t+1)}} \geqslant (\log_{r+1} n)^{\eta(t)} + \frac{p}{2} - 4^{\frac{1}{ \varphi(t) - \eta(t+1)}},
\] or equivalently:
\[
\big(\log_{r+1} s(n, t, p)\big)^{\eta(t)} \geqslant \ell^{\eta(t)} - 1/2.
\]
Following the exact same steps, we end up applying \Cref{lem:log inequality} with $x := \eta(t)$ instead of $\varphi(t)$.
As $c_0$ and $c_1$ are unchanged, we only need to verify that $\eta(t) \leqslant 1 - c_0 - \log_\ell(2c_1)$.
As in the proof of \eqref{eq:recursionf}, we have
$\log_\ell(2c_1) \leqslant 6 \varphi(t-1)$
and $1 - c_0 \geqslant 7 \varphi(t-1)$.
Furthermore, we know that $\eta(t) \leqslant \varphi(t-1)$ (\Cref{def:phigamma}), thus
\[
1 - c_0 - \log(2c_1) - \eta(t) \geqslant 0.
\]
Hence by an application of \Cref{lem:log inequality} we have
\[
\left( \ell - c_1 \cdot \ell^{c_0}\right)^{\eta(t)}  \geqslant
\ell^{\eta(t)} - 1/2.
\]

Thus, $h\big( s(n, t, p), t, p+1\big) \geqslant h(n, t, p)$.

\medskip

\noindent {\bf Proof of \eqref{eq:recursiong}.}
This is the main constraining inequality.
Let $x = \frac{n}{g\left(s(n, t, p), t, p+1\right)}$.
In order to show \eqref{eq:recursiong} we will prove the following equivalent inequality:
\begin{equation}
\label{eq:logx}
\log_{6(r+1)} x \leqslant \log_{6(r+1)} \frac{n}{g(n, t, p)}.
\end{equation}
Recall that (by \hyperref[def:functions]{definition}) 
$
s(n,t,p) \leq g(n/3, t-1,p) \leq n.
$
Therefore
\begin{equation*}
g\big(s(n, t, p), t, p+1\big)
\geqslant \frac{s(n, t, p)}{(6(r+1)) \pow \left(2 \ell^{\gamma(t)} \cdot (3\ell^{\varphi(t)} - p - 1) \right)}.
\end{equation*}
Using \Cref{lem:lowbdstretch} we get
\begin{equation*}
g\big(s(n, t, p), t, p+1\big) \geqslant \frac{n}{(6(r+1)) \pow \big(2  \ell^{\gamma(t-1)} \cdot (3\ell^{\varphi(t-1)} - p) + 2\ell^{\gamma(t)} \cdot (3\ell^{\varphi(t)} - p - 1)  +1 \big)}.
\end{equation*}
So
\begin{equation}
\label{eq:xleq}
x \leqslant (6(r+1)) \pow \left ( 2  \ell^{\gamma(t-1)} \cdot (3\ell^{\varphi(t-1)} - p) + 2\ell^{\gamma(t)} \cdot (3\ell^{\varphi(t)} - p - 1)  +1 \right ).
\end{equation}
Notice that from the \hyperref[def:functions]{definition} of $g$ we have
\begin{equation}
\label{eq:nsurg}
\log_{6(r+1)} \frac{n}{g(n, t, p)} = 2\ell^{\gamma(t)} \cdot (3\ell^{\varphi(t)}-p).
\end{equation}
From \eqref{eq:xleq} and \eqref{eq:nsurg} we obtain
\begin{align*}
\log_{6(r+1)} x - \log_{6(r+1)} \frac{n}{g(n, t, p)}
&\leqslant 1 + 2  \ell^{\gamma(t-1)} \cdot (3\ell^{\varphi(t-1)} - p) + 2\ell^{\gamma(t)} \cdot (3\ell^{\varphi(t)} - p - 1) \\
& \quad - 2(\ell^{\gamma(t)} \cdot (3\ell^{\varphi(t)}-p))\\
&\leqslant 1 + 2 \left ( \ell^{\gamma(t-1)} \cdot (3\ell^{\varphi(t-1)} - p) - \ell^{\gamma(t)} \right )\\
&\leqslant 1 + 2 \left ( 3\ell^{\gamma(t-1) +\varphi(t-1)} -
\ell^{\gamma(t)} \right ) & \hspace{-2cm}\text{as}\ p\geq 0.
\end{align*}
Hence, for \eqref{eq:logx} to be true (i.e., for the left-hand side above to be negative) it is sufficient that $3\ell^{\gamma(t-1) + \varphi(t-1)} - \ell^{\gamma(t)} \leq -1/2$ or, equivalently, that
\begin{equation}\label{eq:demi}
\ell^{\gamma(t-1) + \varphi(t-1)} \left (3 - \ell^{\gamma(t) - \gamma(t-1) - \varphi(t-1)} \right ) \leq -1/2.
\end{equation}
Observe that we always have $\ell^{\gamma(t-1) + \varphi(t-1)}\geq 1$ and that by our assumption~\eqref{cond:bign} on $n$,
\begin{align*}
\ell^{\gamma(t) - \gamma(t-1) - \varphi(t-1)} &\geqslant 2 \pow
    {\frac{\gamma(t) - \gamma(t-1) - \varphi(t-1)}{\varphi(t)}}\\
    &\geqslant 2 \pow \frac{7\varphi(t-1)}{\varphi(t)}& \text{by \Cref{def:phigamma}}\\
    &\geqslant 3.5 & \text{as}\ \varphi(t)< \varphi(t-1)\ \text{by}\ \text{\Cref{def:phigamma}}
\end{align*}

So \eqref{eq:demi} holds. As a consequence \eqref{eq:logx} holds, as desired. This concludes the proof of \eqref{eq:recursiong}.

\medskip

\noindent {\bf Proof of \eqref{eq:middlef}.}
We want to prove
\[
(\log_{r+1} (n/3))^{\varphi(t-1)} - \tfrac{p}{2} - 4^{\frac{1}{ \varphi(t-2) - \eta(t-1)}} \geqslant (\log_{r+1} n)^{\eta(t)} + \tfrac {p}{2} - 4^{\frac{1}{ \varphi(t-2) - \eta(t-1)}}.
\]
Since $p \leqslant 2 \ell^{\varphi(t)}$ and $\log 3 \leqslant 2$, it is sufficient to prove:
\[
\ell^{\varphi(t-1)} \geqslant \ell^{\eta(t)} + 2\ell^{\varphi(t)} + 2.
\]
And since $\varphi(t) \leqslant \eta(t)$ and $\ell^{\varphi(t)} \geqslant 2$, it is sufficient to prove: 
\[
\ell^{\varphi(t-1)} \geqslant 4\ell^{\eta(t)}.
\]
By \Cref{cond:bign}, $\ell \ge 4^{\frac{1}{\varphi(t) \cdot(\varphi(t - 1) - \eta(t))}} \geqslant 4^{\frac{1}{ \varphi(t - 1) - \eta(t)}}$, hence the inequality holds.

\medskip

\noindent {\bf Proof of \eqref{eq:middleg}.}
By \Cref{lem:lowbdstretch}, it is enough to show that $$\frac{n} {(6(r+1)) \pow \left(2 \ell^{\gamma(t-1)} \left( 3\ell^{\varphi(t-1)} - p\right) + 1 \right)}\ge g(n,t,p)=\frac{n}{(6(r+1)) \pow \left(2\ell^{\gamma(t)} \cdot (3\ell^{\varphi(t)} - p)\right)},$$
which is equivalent to $$2 \ell^{\gamma(t-1)} \left( 3\ell^{\varphi(t-1)} - p\right) + 1\le 2\ell^{\gamma(t)} \cdot (3\ell^{\varphi(t)} - p).$$
Observe that the function $p\mapsto 2\ell^{\gamma(t)} \cdot (3\ell^{\varphi(t)}-p)-\left(2 \ell^{\gamma(t-1)} \left( 3\ell^{\varphi(t-1)} - p\right) + 1\right)$ is increasing (since the derivative is equal to $2\ell^{\gamma(t)}-2\ell^{\gamma(t-1)}>0$), and thus it is enough to prove the inequality above when $p=0$, that is:
\begin{eqnarray}\label{eq:pfof(8)}
2 \ell^{\gamma(t-1)}\cdot 3\ell^{\varphi(t-1)} + 1\le 2\ell^{\gamma(t)} \cdot 3\ell^{\varphi(t)}=6 \ell^{\varphi(t)+\gamma(t)}.
\end{eqnarray}
We now remark that since $\ell^{\varphi(t-1)}\ge \ell^{\varphi(t)}\ge 2$ we have
\[
2 \ell^{\gamma(t-1)} \left( 3\ell^{\varphi(t-1)} \right) + 1\le 12\ell^{\varphi(t-1)+\gamma(t-1)},\]
and therefore, the inequality in \eqref{eq:pfof(8)} is implied by \[\ell^{\varphi(t)+\gamma(t)-\varphi(t-1)-\gamma(t-1)}\ge 12/6=2.\]

Recall that by \hyperref[def:phigamma]{definition} we have $\gamma(t)-\gamma(t-1) \geq 8 \varphi(t-1)$, and thus
\[
\varphi(t)+\gamma(t)-\varphi(t-1)-\gamma(t-1)\ge \varphi(t)+7\varphi(t-1)\ge \varphi(t).
\]
It follows that \[\ell^{\varphi(t)+\gamma(t)-\varphi(t-1)-\gamma(t-1)}\ge \ell^{\varphi(t)}\ge 2,\] as desired.
\end{proof}

\end{document}